\documentclass[reqno,12pt]{amsart}
\usepackage{amsfonts,amsmath,amssymb,subeqnarray}

\usepackage[refpage,noprefix]{nomencl}
\usepackage{verbatim}

\usepackage{amsmath} 
\usepackage{amssymb} 
\usepackage[mathscr]{eucal}
\usepackage{graphicx}
\usepackage{graphics}
\usepackage{esint}
\usepackage{color}
\usepackage[normalem]{ulem}

\usepackage{layout}
\def\Xint#1{\mathchoice 
{\XXint\displaystyle\textstyle{#1}}%
{\XXint\textstyle\scriptstyle{#1}}%
{\XXint\scriptstyle\scriptscriptstyle{#1}}%
{\XXint\scriptscriptstyle\scriptscriptstyle{#1}}%
105
\!\int} 
\def\XXint#1#2#3{{\setbox0=\hbox{$#1{#2#3}{\int}$} 
\vcenter{\hbox{$#2#3$}}\kern-.5\wd0}} 
 
\def\dashint{\Xint-} 
 
\def\cP{{\mathcal P}}

\def\cT{{\mathcal T}}

\def\cX{{\mathcal X}} 
 
\def\cL{{\mathcal L}} 
\def\cM{{\mathcal M}} 
\def\cK{{\mathcal K}}

\def\be{\boldsymbol{e}}

\def\a{\alpha}

\def\d{\delta} 
\def\e{\varepsilon} 
\def\l{\lambda} 
 
\def\r{\varrho}

\def\o{\omega} 
\def\O{\Omega}

\def\p{\partial} 
\def\L{\Lambda} 

\newfam\Bbbfam 
\font\tenBbb=msbm10 
\font\sevenBbb=msbm7 
\font\fiveBbb=msbm5 
\textfont\Bbbfam=\tenBbb 
\scriptfont\Bbbfam=\sevenBbb 
\scriptscriptfont\Bbbfam=\fiveBbb

\newcommand{\R}     {\mathbb{R}} 
\newcommand{\Z}     {\mathbb{Z}} 
\newcommand{\N}     {\mathbb{N}}

\newcommand{\E}     {\mathbb{E}}

\newcommand{\Oe}     {\Omega_{\e}} 
\newcommand{\Pe}     {\Pi_{\e}}

\newcommand{\Pie}     {\Pi_{\e}}
\newcommand{\abs}[1]{{\lvert #1\rvert}}
\newcommand{\norm}[1]{{\lVert #1\rVert}}

\newcommand{\dist}{{\operatorname {dist}}} 
\newcommand{\diam}{{\operatorname {diam}}} 
\newcommand{\supp}{{\operatorname {supp}}}

\def\1{{\mathchoice {1\mskip-4mu\mathrm l}      
{1\mskip-4mu\mathrm l} 
{1\mskip-4.5mu\mathrm l} {1\mskip-5mu\mathrm l}}}

\def\co{{\rm c}}

\def\per{{\rm per}}

\def\Xint#1{\mathchoice
    {\XXint\displaystyle\textstyle{#1}}%
    {\XXint\textstyle\scriptstyle{#1}}%
    {\XXint\scriptstyle\scriptscriptstyle{#1}}%
    {\XXint\scriptscriptstyle\scriptscriptstyle{#1}}%
    \!\int}
\def\XXint#1#2#3{{\setbox0=\hbox{$#1{#2#3}{\int}$}
       \vcenter{\hbox{$#2#3$}}\kern-.5\wd0}}

\def\dashint{\Xint-}
 
\newtheorem{theorem}{Theorem}[section]
\newtheorem{proposition}[theorem]{Proposition}
\newtheorem{corollary}[theorem]{Corollary}

\newtheorem{lemma}[theorem]{Lemma}

\newtheorem{remark}[theorem]{Remark}

\newenvironment{proofsect}[1] 
{\vskip0.1cm\noindent{\bf #1.}\hskip0.5cm}

\numberwithin{equation}{section}

\begin{document}

\renewcommand{\labelenumi}{{(\roman{enumi})}}

\title{Nonlinear elastic free energies and  gradient Young-Gibbs measures}

\author[Roman Koteck\'{y} and Stephan Luckhaus]{}
\maketitle

\thispagestyle{empty} 
\vspace{0.2cm} 

\centerline
{\sc
Roman Koteck\'{y},\footnote{
  Charles University, Prague, Czech Republic,
  and University of Warwick, United Kingdom,
  {\tt R.Kotecky@warwick.ac.uk}}
\/
Stephan Luckhaus\footnote{
  Institut f\"ur Mathematik, Leipzig University,  D-04009 Leipzig, Germany,\\
  {\tt Stephan.Luckhaus@math.uni-leipzig.de}}
}

\vspace{0.4cm}


\begin{abstract}
We investigate, in a fairly general setting, the limit of large volume equilibrium Gibbs measures for elasticity
type Hamiltonians with clamped boundary conditions.
The existence of a quasiconvex free energy, forming the large deviations rate functional, is shown using a new interpolation lemma for partition functions. The local behaviour of the Gibbs measures can be parametrized by Young measures on the space of gradient Gibbs measures. In view of unboundedness of the state space, the crucial tool here is an exponential tightness estimate that holds for a vast class of potentials and the construction of suitable compact sets of gradient Gibbs measures.
\end{abstract}

\section{Setting and results}
 
The aim of the paper is to derive, in a mathematically rigorous way, macroscopic elasticity with variational principles formulated in terms of nonlinear elastic free energy from
equilibrium statistical mechanics with gradient Gibbs measure on the space of displacements of individual atoms, based on a microscopic Hamiltonian.

We begin with the microscopic description. In general, we consider the space of microscopic configurations $X:\Z^d\to\R^m$.
This  includes the case of elasticity where we actually have $m=d$ with $X(i)$ denoting the vector of displacement of the atom
labelled by $i$ as well as the case of random interface with $m=1$ and $X(i)$ denoting the height of interface above the lattice site $i$.

For  any fixed $Y:\Z^d\to\R^m$ and any finite $\Lambda\subset\Z^d$, the Gibbs measure $\mu_{\Lambda,Y}(d X)$
on $(\R^m)^\Lambda$ under the boundary conditions $Y$ is defined in terms of a Hamiltonian $H$ with a finite range interaction $U$.
Namely, let a finite $A\subset \Z^d$   and a  function $U:( \R^m)^{A}\to\R$  be given. We use  $R_0=\diam A$ to denote the range of potential $U$. We also assume that $U$ is invariant under rigid motions
(i.e.  $U(\bold{R}(\tau_a X))=U(X)$ for any $X\in  ( \R^m )^{A}$ and any $\bold{R}\in SO(m)$, $a\in \R^m$, with
 $\bold{R}(\tau_a X)(i)= \bold{R}(X(i)+a)$). In addition, suitable growth conditions on $U$ will be specified later and, 
 for simplicity (and without loss of generality), we suppose that  $\{0,\pm \be_1, \dots, \pm\be_d\}\subset A$, where $\be_1, \dots, \be_d$ are the unit vectors in coordinate directions.
Using  $X_A$ to denote the restriction of $X$ to $A$ for any $X:\Z^d\to\R^m$ and any $A\subset\Z^d$, we introduce  the Hamiltonian
\begin{equation}
\label{E:H}
H_{\Lambda}(X)= \negthickspace\negthickspace\sum_{j\in\Z^d\colon \tau_j(A)\subset \Lambda } 
\negthickspace\negthickspace U(X_{\tau_j(A)})
\end{equation}
with $\tau_j(A)=A+j=\{i\colon i-j\in A\}$.
The corresponding Gibbs measure on $(\R^m)^\Lambda$ (equipped with the corresponding Borel $\sigma$-algebra) 
is defined as
\begin{equation}
\label{E:mu}
\mu_{\Lambda,Y}(d X)=\frac{\exp\bigl\{- \beta H_{\Lambda}(X)\bigr\}}{Z_{\L,Y}}
\1_{\L,Y}(X)\prod_{i\in \Lambda} d X(i)
\end{equation}
with
\begin{equation}
\label{E:Z}
Z_{\L,Y}=\int_{(\R^m)^\Lambda}\exp\bigl\{- \beta H_{\Lambda}(X)\bigr\}\1_{\L,Y}(X)\prod_{i\in \Lambda} d X(i).
\end{equation}
Here, we introduce  boundary conditions by considering a fixed configuration $Y$ in the boundary layer  
\begin{equation}
S_{R_0}(\L)=\{i\in\L |  \dist(i,\Z^d\setminus \L)\le R_0\}
\end{equation}
and restricting the configurations $X$ to the set
\begin{equation}
\label{E:Nsoft}
\{ X\in (\R^m)^\Lambda:   \abs{X(i)- Y(i)}<1  \text{ for all } i\in S_{R_0}(\L)\}
\end{equation}
with the indicator $\1_{\L,Y}(X)$.
In the following we consider the inverse temperature $\beta$ to be incorporated in the Hamiltonian and skip it from the notation. 

In the standard setting of elasticity theory, we are interested in the macroscopic equilibrium configuration in an open set $\Omega\subset\R^d$ 
under fixed boundary conditions $u\colon \p\Omega\to \R^m$. To link this with the microscopic description, we superimpose a finite lattice $\Oe$
over $\Omega$.
Namely, for any ${\e}\in(0,1)$, let 
\begin{equation}
\Oe=\tfrac1{\e}\bigl({\e}\Z^d\cap\Omega\bigr)\equiv \Z^d\cap\tfrac1{\e}\Omega.
\end{equation}
Naturally,  $\tfrac1{\e}\Omega$ and ${\e}\Z^d$ denotes the rescaling of $\O$ and $\Z^d$ by $\tfrac1{\e}$ and $\e$, respectively.
We will assume certain regularity of the boundary $\p \O$ of the domain $\O$. Namely, using $\p_{\r}\O$
to denote the intersection of the $\r$-neighbourhood of the boundary $\p\O$ with $\O$, 
\begin{equation}
\p_{\r}\O=\{x\in\O:\dist(x,\p\O)<\r\},\  \r>0, 
\end{equation}
we assume that $\O$
is a domain with Lipschitz boundary that allows to check the following condition: 
\smallskip

\noindent
\textbf{(A${}_{\boldsymbol{\partial}}$)}
\textit{There exist constants   $\rho_0$, $\e_0$, and $C_\p$    such that,  for any  $\r\le\rho_0$ and $\e\le\e_0$,  the number of points  in  the strip $ S_{\rho/\e}=(\p_{\r}\O)_\e=\{i\in \Oe\mid i\e\in \p_{\r}\O\} $ 
is bounded as $\abs{S_{\rho/\e}}\le \e^{-d} C_{\p}  \abs{\p\O} \r$.}
\smallskip

Further, for any  $u\in L_{1,\text{loc}}(\R^d,\R^m)$,  let $X_{u,\e}: \Z^d\to\R^m$ be defined by  
\begin{equation}
X_{u,\e}(i)= \frac1{\e} \dashint_{\varepsilon i + Q(\e)} u(y)\, dy
\end{equation}
for any $i\in\Z^d$. Here, $Q(\e)=[-\tfrac{\e}2,\tfrac{\e}2]^d$ and $\dashint$ denotes the mean value.
On the other hand,  let 
\begin{equation}
\Pie: (\R^m)^{\Z^d}_0\to  W^{1,p}(\R^d)
\end{equation}
be a canonical  interpolation $X\to v$ such that $v({\e} i)={\e} X(i)$ for any $i\in\Z^d$.
Here $(\R^m)^{\Z^d}_0$ is the set of functions $X:\Z^d\to\R^m$ with finite support.
To fix the ideas, we can consider a triangulation of $\Z^d$ into simplices with vertices in $\e\Z^d$ and choose $v$ on each simplex as the linear interpolation
of the values ${\e} X(i)$ on the vertices $\e i$. 

Our main task is to study the asymptotic behaviour (with $\e\to 0$) of the  measure
$\mu_{\Oe, X_{u,\e}}(d X )$
in terms of  minimizers $v$ of the functional $\int_{\Omega} W(\nabla v(x)) dx$ with the boundary condition $v=u$ on $\p\Omega$.
It turns out that  the free energy $W(L)$ featuring in the above integral is defined, for any affine function $L\colon \R^d\to \R^m$,  by the limit
\begin{equation}
W(L)=-\lim_{\e\to 0} \e^d|\O|^{-1} \log  Z_{\Oe,L},
\end{equation}
where $Z_{\Oe,L}$ is a shorthand for $Z_{\Oe,X_{L,\e}}$ reflecting the fact that, with an affine function $L$, we actually have $X_{L,\e}(i)=L(i)$
with the condition $\abs{X-L}_{S_{R_0}(\Oe),\infty}\le 1$ reading  $\abs{\Pie(X)-L}_{S_{R_0/\e}(\O),\infty}\le \e$.
For the existence of the limit, see Proposition~\ref{T:WL} below.

Using $\nabla X(i)=( \nabla_{1} X(i), \dots, \nabla_{d} X(i))$ for the discrete gradient, $\nabla_{k} X(i)=X(i+\be_k)-X(i)$,  $k=1,\dots,d$, and defining $\abs{\nabla X(i)}^p=\sum_{k=1}^d \abs{\nabla_k X(i)}^p$,
our main assumptions are the following basic restrictions on the growth of the potential $U$ (part of the lower bound is the boudedness from below that can be stated, without loss of generality, as an assumption of non-negativity):
\smallskip 

\begin{list}{}{\setlength{\topsep}{0cm}\setlength{\leftmargin}{0.75cm}\setlength{\parsep}{0cm}\setlength{\itemsep}{-0.01cm}}
\item[\textbf{(A1)}]
\textit{There exist constants $ p>0$ and  $c\in (0,\infty)$  such that 
\begin{equation*}
U(X_A)\ge c\abs{\nabla X(0)}^p
\end{equation*}
for any    $X\in (\R^m)^{\Z^d}$. }
\smallskip
\item[\textbf{(A2)}]
\textit{There exist  constants  $ r>1$  and  $C\in(1,\infty)$  such that 
\begin{equation*}
\phantom{xxxx}U(s X_A+(1-s) Y_A + Z_A)\le C\bigl(1+U( X_A) + U(Y_A) +\sum_{i\in A} \abs{Z(i)}^{r}\bigr)
\end{equation*}}
\textit{for any $s\in[0,1]$ and any   $X,Y,Z\in (\R^m)^{\Z^d}$.}
\end{list}
\smallskip

Increasing possibly the constant $C$ to incorporate the term $U(0)$ from (A2) applied with $s=1$ and $Y=0$, we have a particular useful implication of (A2) in the form
\begin{equation}
\label{E:A2short}
U(X_A)\le C\bigl(1+ U(Z_A)+\sum_{i\in A} \abs{X(i)-Z(i)}^{r}\bigr).
\end{equation}

\begin{remark}
\label{R:gradiendbound}
In view of the invariance of the function $U$ under rigid motions, it actually depends only on gradients 
$\nabla X(i)$, $i\in A$. With the help of discrete Poincar\'e inequality, 
the condition (A2) implies
\begin{equation}
\label{E:A2grad}
U(X_A)\le C\bigl(\sum_{i\in Q(R_0)} \abs{\nabla X(i)}^r+ 1\Bigr),
\end{equation}
with a suitable constant $C$ (again not necessarily the same as that in (A2)) and $Q(R_0)\subset \Z^d$  denoting a cube  (of side $R_0=\diam A$) containing $A$.
\end{remark}

\subsection{Free energy}

A prerequisite to our main statements is the existence of the free energy as a function of the the affine deformation $L$ and its 
continuity and quasiconvexity.

\begin{proposition}[Existence of the free energy]\hfill
\label{T:WL}

\noindent
Suppose that (A2) holds with $r\ge1$. Then,
for any affine $L:\R^d\to\R^m$, the limit 
\begin{equation}
\label{E:W}
W(L)=-\lim_{\e\to 0} \e^d \abs{\O}^{-1} \log Z_{\Oe,L} 
\end{equation}
exists and does not depend on $\O$.
\end{proposition}
\begin{remark}
Instead of the condition (A2), it is enough here to assume that $U(X_{A})$ is bounded by a fixed constant for any  $X$ such that $\abs{X(i)-L(i)} \le 1$ for all  $i\in A$.
\end{remark}
\begin{proof}
The existence of the limit and its independence on $\O$
follows easily by  the standard methods with the help of an approximative subadditivity 
(of $-\log  Z_{\Lambda,L}$): 
if  $\L\subset\Z^d$ is finite and 
$\L_1$ and $\L_2$ are its disjoint subsets,  $\L_1\cup\L_2= \L$,
then 
\begin{equation}
\label{E:subaditivity}
\log Z_{\L,L}  \ge \log Z_{\L_1,L} + \log Z_{\L_2,L} -B(L) \abs{S(\L_1,\L_2)},
\end{equation}
where $B(L)=C(1+C+ (1+C d \norm{L}^r)R_0^d) $ is  a constant
and 
\begin{equation}
S(\L_1,\L_2)=\Lambda\cap(\Lambda_1)_{R_0}\cap (\Lambda_2)_{R_0},
\end{equation}
where $(\Lambda_k)_{R_0}$, $k=1,2$, is the $R_0$-neighbourhood of $\Lambda_k$.
Indeed,  using the the assumption (A2) (resp. its implication \eqref{E:A2short}
), we have
\begin{equation}
U(X_{\tau_j(A)})\le C\bigl(1+ U(L_{\tau_j(A)})+R_0^d\bigr),
\end{equation}
for all $j$ such that  $\tau_j(A)\subset\L$ and in the same time  $\tau_j(A)\cap \L_1\neq \emptyset$ as well as $\tau_j(A)\cap \L_2\neq \emptyset$ (which implies $\tau_j(A)\subset S(\L_1,\L_2)$). 
Hence, 
\begin{equation}
\label{E:L1capL2}
H_{\L}(X)\le H_{\L_1}(X)+H_{\L_2}(X)+B(L) \abs{S(\L_1,\L_2)}
\end{equation}
for any $X$ satisfying $\1_{\L_1,L}(X)\1_{\L_2,L}(X)=1$.
Thus, restricting first the range of integration in the definition of $Z_{\Lambda,L}$ by inserting  the indicator $\1_{\L_1,L}(X)\1_{\L_2,L}(X)$ (notice that $\1_{\L_1,L}(X)\1_{\L_2,L}(X)\le \1_{\L,L}(X)$) and using then the inequality \eqref{E:L1capL2}, we get the claim.
\end{proof}

\begin{proposition}[Quasiconvexity of the free energy]\hfill
\label{P:QC}

\noindent
Assume that  $U$  satisfies the assumptions (A1) and (A2) with $r, p \ge1$,  $\frac1r>\frac1p-\frac1d$.
The free energy $W(L)$ is continuous with $r$ growth,
\begin{equation}
\abs{W(L)}\le \overline C(1+\norm{L}^r),
\end{equation} 
it is quasiconvex, and, as a consequence,
$\int_{\O} W(\nabla v(x))dx $ is a weakly lower semicontinuous functional on $W^{1,r}(\O)$.
\end{proposition}
\begin{remark}
\label{R:non-convex}
Notice that, in general, for $m\ge 2$, we should not expect that the free energy is convex. We  provide an explicit class of examples in Section~\ref{S:nonconv}.
\end{remark}

The proof of this and the remaining statements in this section is deferred to the next section as it hinges on the crucial Interpolation Lemma and Exponential Tightness Lemma
to be stated there (and  proven in the Appendix).


\subsection{Large deviations}

Before formulating the theorem whose consequence is the  large deviations principle for the measure  $\mu_{\Oe,X_{u,\e}}(dX)$,
we introduce several restricted partition functions. For any finite $\L\subset\Z^d$  and any set $\mathcal S\subset (\R^m)^{\L}$, we  write
\begin{equation}
Z_{\L}(\mathcal S)= \int_{\mathcal S}\exp\{- H_{\L} (X)\} \prod_{i\in \L}dX(i).
\end{equation}
In particular,  for any  $Y\in (\R^m)^{\Z^d}$, 
\begin{equation}
Z_{\L,Y}= Z_{\L}({\mathcal N}_{\L,R_0,\infty}(Y)),
\end{equation}
where ${\mathcal N}_{\L,R_0,\infty}(Y)$ is the set corresponding to the indicator $\1_{\L,Y}$,
\begin{equation}
\label{E:Nsoft}
{\mathcal N}_{\L,R_0,\infty}(Y)= \{ X:\L \to \R^m |   \abs{X(i)- L(i)}<1  \text{ for all } i\in S_{R_0}(\L)\}.
\end{equation}

Now,  for any $v\in  L^r(\O)$, consider the   neigbourhood
\begin{equation}
\label{E:Nrd}
{\mathcal N}_{\Oe,r}(v,\kappa)= \{ X:\Oe \to \R^m |   \norm{\Pie(X)-v}_{L^{r}(\O)} <\kappa\abs{\O}^{\frac1r+\frac1{d}}\}
\subset (\R^m)^{\Oe}
\end{equation}
with the corresponding partition function 
\begin{equation}
\label{Zr,kappa}
Z_{\Oe}({\mathcal N}_{\Oe,r}(v,\kappa))= \int_{{\mathcal N}_{\Oe,r}(v,\kappa)}\exp\{- H_{\Oe} (X)\} \prod_{i\in\Oe}dX(i).
\end{equation}
Notice that  the volume dependent factor $\abs{\O}^{\frac1r+\frac1{d}}$  is chosen so that the set 
${\mathcal N}_{\Oe,r}(v,\kappa)$ does not change under the rescaling $\e\to\sigma\e$ and $\O\to \sigma \O$ with
 $v(x)\to \sigma v(\frac{x}\sigma)$.
It is  easy to verify the following equivalences of the corresponding norms (uniformly in $\e$), 
\begin{equation}
\label{E:discontineq}
\tfrac12\norm{\Pie(X)-v}^r_{L^{r}(\O)}\le \e^{d+r}\sum_{i\in\Oe}\abs{X(i)-X_{v,\e}(i)}^r \le 2 \norm{\Pie(X)-v}^r_{L^{r}(\O)}
\end{equation}
and 
\begin{equation}
\tfrac12\norm{\nabla(\Pie(X)-v)}^p_{L^{p}(\O)}\le \e^{d}\sum_{i\in\Oe}\abs{\nabla X(i)-\nabla X_{v,\e}(i)}^p \le 2 \norm{\nabla(\Pie(X)-v)}^p_{L^{p}(\O)}.
\end{equation}
Using $\chi_\e$ to denote the characteristic function of $Q(\e)$ and comparing piecewise linear and piecewise constant interpolation, we get
\begin{equation}
\Vert \Pie(X)-\sum_i \e \chi_\e(\cdot -\e i) X(i) \Vert^r_{L^{r}(\O)}\le 2 \e^{d+r}\sum_{i\in\Oe}\abs{\nabla X(i)}^{\a}\abs{X(i)}^{r-\a} 
\end{equation}
for any $\a<1$,
which by the Sobolev estimates entails
\begin{equation}
\Vert \Pie(X)-\sum_i \e \chi_\e(\cdot -\e i) X(i)\Vert^r_{L^{r}(\O)}\le 2 \e^{\a}  \norm{\Pie(X)}_{W^{1,p}}^r,
\end{equation}
where $\a=\min(1,d+r-\frac{dr}p)$.

In view of \eqref{E:discontineq}, the condition
$\norm{\Pie(X)-v}^r_{L^{r}(\O)}<\kappa^r\abs{\O}^{1+\frac{r}{d}}$ is, up to a change of $\kappa$ multiplying it by a fixed factor,  equivalent to 
$\sum_{i\in\Oe}\abs{X(i)-X_{v,\e}(i)}^r<\kappa^r\abs{\Oe}^{1+\frac{r}{d}}$.
This suggests the notation
\begin{equation}
\label{E:NrdL}
\overline{\mathcal N}_{\L,r}(Z,\kappa)= \{ X:\L \to \R^m |   \norm{X-Z}_{\ell^r(\L)}<\kappa\abs{\L}^{\frac1r+\frac1{d}}\}
\end{equation}
for any $\Lambda\subset\Z^d$ and $Z\in (\R^m)^{\Lambda}$.
As observed above, for $\L=\Oe$ and $Z=X_{v,\e}$, the sets ${\mathcal N}_{\Oe,r}(v,\kappa)$
and $\overline{\mathcal N}_{\Oe,r}(Z,\kappa)$ actually differ only by  change  of $\kappa$ up to the factor 2, $\kappa\to 2\kappa$.

\begin{theorem}\hfill
\label{T:LDab}

\noindent
Assume that  $U$  satisfies the assumptions (A1) and (A2) with with $r\ge p >1$,  $\frac1r>\frac1p-\frac1d$ and let $v\in W^{1,p}(\O)$. Further, let 
 \begin{equation}
\label{E:Fkappa}
F_{\kappa,\e}(v)=- \e^d \abs{\O}^{-1} \log Z_{\Oe}({\mathcal N}_{\Oe,r}(v,\kappa)),
\end{equation}
and 
\begin{eqnarray}
&F_{\kappa}^+(v)=\limsup_{\e\to 0} F_{\kappa,\e}(v)\\
&F_{\kappa}^-(v)=\liminf_{\e\to 0} F_{\kappa,\e}(v)
\end{eqnarray}
Then:

\noindent
a) $ \lim_{\kappa\to 0} F_{\kappa}^-(v)\ge \frac1{\abs{\O}}\int_{\O} W(\nabla v(x)) dx$.

\noindent
b) If $v\in W^{1,r}(\O)$ then $ \lim_{\kappa\to 0} F_{\kappa}^+(v)\le \frac1{\abs{\O}}\int_{\O} W(\nabla v(x)) dx$.
\end{theorem}

As a consequence, we get the following large deviation behaviour for $\mu_{\Oe, X_{u,\e}}(dX)$.
For convenience, we actually extend the measure $\mu_{\Oe, X_{u,\e}}(dX)$ defined on $(\R^m)^{\Oe}$
to $(\R^m)^{\Z^d}$ by 
defining $\mu_{\e,u}= \mu_{\Oe, X_{u,\e}}\times \Pi_{i\in \Z^d\setminus\Oe}\delta_{X_{u,\e}(i)}$ and adding the assumption  that the function $u$ is supported on a bounded set in $\R^d$.

\begin{theorem}[Large deviation principle]\hfill
\label{T:LD}

\noindent
Assume that  $U$  satisfies the assumptions (A1) and (A2)  with $r\ge p >1$,  $\frac1r>\frac1p-\frac1d$,  and let $u\in W^{1,p}(\O)$. 
If $p=r$ or, more generally,
\begin{equation}
\label{E:Lav}
\int_{\O} W(\nabla v(x))dx= \sup_{\delta>0} \inf\Bigl\{ \int_{\O} W(\nabla \bar v(x))dx\, \big\vert\, \bar v \in W^{1,r}_0(\O)+u, \norm{\bar v-v}_{L^r}<\delta \Bigr\},
\end{equation}
for every $v\in W^{1,p}_0(\O)+u$, then the Gibbs measures $\mu_{\e,u}(dX)$  satisfies the large deviation principle with the rate $\e^{-d}$ and  the rate functional 
\begin{equation}
I(v)=\int_{\O} W(\nabla v(x))dx - \min_{\bar v\in W^{1,r}_0(\O)+u}\,  \int_{\O} W(\nabla \bar v(x))dx.
\end{equation} 
Namely:

\noindent
a) For any $C\subset W^{1,p}_0(\O)+u$ closed in the weak topology, we have
\begin{equation}
\limsup_{\e\to 0} \e^d \log \mu_{\e,u}(\Pie^{-1}(C)) \le - \inf_{v\in C} I(v).
\end{equation}

\noindent
b) For any $O\subset W^{1,p}_0(\O)+u$ open in the weak topology, we have
\begin{equation}
\liminf_{\e\to 0} \e^d \log \mu_{\e,u}(\Pie^{-1}(O)) \ge - \inf_{v\in O} I(v).
\end{equation}
\end{theorem}

\begin{remark}
Similar statements hold for  Neumann or periodic boundary conditions.

The existence of $v\in W^{1,p}_0(\O)+u$ such that there is a strict inequality in \eqref{E:Lav}
is the  so called Lavrentiev gap. 
Thus, our claim about LDP comes under the assumption of the absence of Lavrentiev gap.
\end{remark}

\begin{remark}
Suppose that $F: W^{1,p}(\O) \to \R$ is weakly continuous and has a growth strictly smaller then $p$. Then the rate functional corresponding to  $H+F\circ \Pie$ 
(i.e. the rate functional of the measure  $\frac{e^{F(\Pie(X))}  \mu_\e(d X)}{\int  e^{F(\Pie(Y)) } \mu_\e(d Y)}$) is 
\begin{equation}
\int_{\O} W(\nabla v(x))dx +F(v) - \min_{\bar v\in W^{1,r}_0(\O)+u}\,  \Bigl[\int_{\O} W(\nabla \bar v(x))dx+F(\bar v)\Bigr].
\end{equation}
\end{remark}
\bigskip

We note that large deviation principle has been discussed   before  (often under more restrictive conditions on the potential and  in the scalar case $m=1$). See, for example, \cite{DGI} (the case of strictly convex potentials and $m=1$) and \cite{S} (potentials satisfying analog of our assumptions (A1) and (A2) with $r=p$ and with large deviation formulated in detail only for $m=1$).

While the large deviation principle clarifies the role of minimizers of the functional $I$ for  the description of the ``macroscopic'' asymptotic behaviour of measures $\mu_{\e,u}$,
our final claim, introducing the notion of ``Young-Gibbs'' measures, inspects the asymptotic ``microscopic'' behaviour of $\mu_{\e,u}$.

\subsection{Gradient Young-Gibbs measures}

Here, we will need the Gibbs specification $\mu_{\Lambda}(d X |  Y)$ in the volume $\Lambda\subset \Z^d$ and with boundary condition $Y$.
First, for any $X,Y\in(\R^m)^{\Z^d}$, we introduce  the Hamiltonian
\begin{equation}
\label{E:H}
H_{\Lambda}(X |  Y)= \negthickspace\negthickspace\sum_{j\in\Z^d\colon \tau_j(A)\cap \Lambda \neq\emptyset} 
\negthickspace\negthickspace U((X\vee Y)_{\tau_j(A)}), \text{ where }  X\vee Y=
\begin{cases} X \text{ on } \Lambda\\  Y \text{ otherwise. }\end{cases}
\end{equation}
Then,
\begin{equation}
\label{E:mu}
\mu_{\Lambda}(d X |  Y)=\frac{\exp\bigl\{-  H_{\Lambda}(X |  Y)\bigr\}}{Z_{\L}(Y)}\prod_{i\in \Lambda} d X(i)\prod_{i\in \L^{\text{\rm c}}} \delta_{Y(i)}(dX(i))
\end{equation}
with
\begin{equation}
\label{E:Z}
Z_{\L}(Y)=\int_{(\R^m)^\Lambda}\exp\bigl\{-  H_{\Lambda}(X |  Y)\bigr\}\prod_{i\in \Lambda} d X(i).
\end{equation}
 A gradient Gibbs measure (with potential $U$) is any probability measure $\mu$ on the space $S=(\R^m)^{\Z^d}/\R^m$ 
 such that 
 \begin{equation}
 \mu(f)= \int \mu_{\Lambda}(f |  Y)\mu(dY)
 \end{equation}
for any finite $\L\subset \Z^d$ and any measurable function $f$ on $S$ (i.e. a measurable function on $(\R^m)^{\Z^d}$ invariant under translations in $\R^m$). Notice that for such $f$, the function
$Y\to  \mu_{\Lambda}(f |  Y)$ is also invariant under translations in $\R^m$ and can thus be integrated with the probability measure $\mu$ on $S$.
We use $\mathcal G$ to denote the set of  all gradient  Gibbs measures
and  $\mathcal G^p$ to denote the set of
all $\mu\in \mathcal G$ with finite $p$th moment.
Both are subsets of the set $\cP(S)$ of probability measures on the space $S$.

Further, we introduce two measure spaces  on sets of measures,  both with the corresponding weak topology:
the space $\mathcal P(\mathcal G^p)$  of probability measures on $\mathcal G^p$ and
the space $\mathcal P(W^{1,p}_0(\O)+u)$  of probability measures on $W^{1,p}_0(\O)+u$ (nonnegative functionals on bounded weakly continuous functions  on $W^{1,p}_0(\O)+u$ of mass 1).
A probability measure  $\gamma\in\mathcal P(W^{1,p}_0(\O)+u)$   together with a map $\nu: \O\times (W^{1,p}_0(\O)+u) \to 
\mathcal P(\mathcal G^p)\subset BC(BC(S)^\ast)^\ast$ 
that is weakly measurable with respect to  $\gamma\times\lambda$, where $\lambda$ is the normalized Lebesgue measure on $\O$,  is called a \emph{gradient Young-Gibbs measure}. 

We will show that the measures $\mu_{\e,u}$ converge, in appropriate sense, to 
a gradient Young-Gibbs measure with a slope condition linking the slope $\nabla v(x)$
with the expectation $\int \E_{\mu}(\nabla X(0))\,d \nu_{x,v}(\mu)$.
 
We formulate the convergence in terms of appropriate test functions.
Namely, we consider the space $\mathcal X$ of  \emph{test functions}
\begin{equation}
\varphi: (W^{1,p}_0(\O)+u)\times BC(S)^\ast_+\times
 \O\to \R
\end{equation}
that are weakly continuous and for any $\delta$ fulfill the growth condition
\begin{equation}
\label{E:PhiLip}
\abs{\varphi(v,\mu,x)}\le \eta(x)\exp\{ c_{\varphi} \norm{v}_{1,p}^p)\}  \mu\bigl(\delta\textstyle{\sum}_{i\in\Lambda}\abs{\nabla X(i)}^p +C(\delta)\bigr)
\end{equation}
with a fixed $\Lambda$, $\eta\in C_0(\O)$ and the constant $c_{\varphi}$ depending only on $\varphi$.

\begin{theorem}[Convergence to gradient Young-Gibbs measures]\hfill
\label{T:GYG}

\noindent
Assume that  the potential $U$  satisfies the assumptions (A1) and (A2) with $r\ge p> 1$,  $\frac1r>\frac1p-\frac1d$ and let $u\in W^{1,p}(\O)$.
There exists a sequence $\e_n\to 0$ and a gradient Young-Gibbs measure $\nu$ with $\gamma$ supported,
if the functional $I$ has no Lavrentiev gap,  on minimisers of $I$ in $W^{1,p}_0(\O)+u$ such that   
\begin{multline}
\label{L:GYGdesint}
\lim_{\L\to\Z^d}\lim_{n\to \infty}\mu_{\e_n,u}\Bigl(\textstyle{\int_{\O}} \varphi(\Pie(X), \mu_{\L}(\cdot|\tau_{\lfloor{x}/{\e}\rfloor}(X)), x)dx\Bigr)=\\=
\int_{W^{1,p}_0(\O)+u}\int_{\O} \nu_{x,v}\bigl(\varphi(v,\cdot,x) 
\bigr)d\gamma(v) dx
\end{multline}
for all $\varphi\in \mathcal X$ and
\begin{equation}
\label{L:GYGslope}
\int \E_{\mu}(\nabla X(0))\,d \nu_{x,v}(\mu)=\nabla v(x)
\end{equation}
for almost all $x$ and almost all  $v$ (with respect to $\gamma$).
\end{theorem}

The idea of of gradient Young-Gibbs measures has its precursor in the non-stochastic case, 
where the Young measures and the methods of $\Gamma$-convergence were often used  in the context of nonlinear elasticity  (see \cite{B} for a review). We were also inspired 
 by the procedure of two-scale convergence that is used in homogenization. See, for example, \cite{BLM}. 
Whenever there is a unique Gibbs measure corresponding to the affine mapping $L=\nabla v(x)$, the measure $\nu_{x,v}$
is actually a Dirac measure. For the scalar case, $m=1$, the unicity of Gibbs measure corresponding to a fixed slope  has been proven for a class of strictly convex potentials \cite{FS}. Notice, however, that even in the scalar case, this is not always the case
as phase transitions may occur \cite{BK}.

\section{Proofs}
\label{S:proofs}

\subsection{Exponential Tightness}
Here, we first state two crucial Lemmas (with the proofs deferred to the Appendix) and then, using them, we prove the claims from Section 1.
In the following, we consider $m,d$, and $\Omega$ to be fixed (often without explicitly mentioning the dependence of various constants on these parameters).  

The first technical Lemma assures a needed tightness of finite volume Gibbs measures
when conditioned on the neighbourhood ${\mathcal N}_{\Oe,r}(v,\kappa)$.
Let, for any $K\in (0,\infty)$, the set ${\mathcal M}_K$ 
be  defined by
\begin{equation}
{\mathcal M}_K= \{ X:\Oe \to \R^m |    H_{\Oe}(X )> K \abs{\Omega_{\e}}\}.
\end{equation}

\begin{lemma}[Exponential Tightness]\hfill
\label{L:tight}

\noindent
Assume that  $U$  satisfies the assumption (A1).   There exists  a constant $K_0$ and,  for any
$r$, and $\kappa$,  a constant  $\e_0=\e_0(r,\kappa)$ such that, for any
 $\e\le \e_0$, $K\ge K_0$,    $u\in L_{1,\text{loc}}(\O,\R^m)$, and $v\in L^r(\O)$,  we have
 \begin{equation}
\label{E:intM^b}
Z_{\Oe}({\mathcal M}_K\cap {\mathcal N}_{\Oe,R_0,\infty}(X_{u,\e}))   \le e^{-\tfrac12 K \abs{\Oe}}
D^{\abs{\Oe}} 
\end{equation}
and
 \begin{equation}
\label{E:intM^f2}
Z_{\Oe}({\mathcal M}_K\cap {\mathcal N}_{\Oe,r}(v,\kappa))   \le
e^{-\tfrac12 K \abs{\Oe}} D^{\abs{\Oe}}  
\end{equation}
with  $D=2(\tfrac2{ c})^{\frac{m}p}c(p,m)$,
where  
$c(p,m)=\int_{\R^m} \exp\bigl(-\abs{\xi}^p\bigr)d\xi $.
\end{lemma}

\begin{remark} 
\label{R:tight}
We refer to the above claim as the exponential tightness since, under additional assumption (A2),  it
 implies that 
  \begin{multline}
\label{E:tight}
\mu_{\e,u}({\mathcal M}_K)=\frac{Z_{\Oe}({\mathcal M}_K\cap  {\mathcal N}_{\Oe,R_0,\infty}(X_{u,\e})}{ Z_{\Oe}({\mathcal N}_{\Oe,R_0,\infty}(X_{u,\e}))}  \le\\
\le e^{-\frac12 K \abs{\Oe}}
(D/\omega(m))^{\abs{\Oe}}  \exp\bigl(C \bigl(H_{\Oe}(X_{u,\e})+(1+R_0^d)\abs{\Oe}\bigr)\bigr)
\end{multline}
with  $\o(m)$ denoting the volume of  a unit ball in $\R^m$. 
Indeed, considering  the   $\ell^{\infty}(\Oe)$-neighbourhood
\begin{equation}
\label{E:Nd}
{\mathcal N}_{\Oe,\infty}(Z)= \{ X:\Oe \to \R^m |   \abs{X(i)- Z(i)}<1 \text{ for all } i\in \Oe\},
\end{equation}
and using (A2) in the form \eqref{E:A2short}, we get
\begin{equation}
\label{E:ZOe>}
Z_{\Oe}({\mathcal N}_{\Oe,R_0,\infty}(X_{u,\e}))\ge Z_{\Oe}({\mathcal N}_{\Oe,\infty}(X_{u,\e})) \ge
\exp\bigl(-C \bigl(H_{\Oe}(X_{u,\e})+(1+R_0^d)\abs{\Oe}\bigr)\bigr)\omega(m)^{\abs{\Oe}}.
\end{equation} 

Similarly, observing that for any $X\in {\mathcal N}_{\Oe,\infty}(Z)$ and a sufficiently small $\e$, we have  $\abs{\Oe}^{1/d}>2/\kappa$ and thus
\begin{equation}
\label{E:L^r-l^r}
\sum_{i\in\Oe}\abs{X(i)-Z(i)}^r\le \abs{\Oe} \le (\tfrac{\kappa}2)^r\abs{\Oe}^{1+\frac{r}{d}},
\end{equation}
 implying ${\mathcal N}_{\Oe,\infty}(Z) \subset  \overline{\mathcal N}_{\Oe,r}(Z,\kappa/2)\subset {\mathcal N}_{\Oe,r}(v,\kappa)$, we get
 \begin{equation}
\label{E:tight}
\frac{Z_{\Oe}({\mathcal M}_K\cap {\mathcal N}_{\Oe,r}(v,\kappa)) }{Z_{\Oe}({\mathcal N}_{\Oe,r}(v,\kappa))}  \le e^{-\frac12 K \abs{\Oe}}
(D/\omega(m))^{\abs{\Oe}}  \exp\bigl(C \bigl(H_{\Oe}(Z)+(1+R_0^d)\abs{\Oe}\bigr)\bigr)
\end{equation}
for any  $Z\in {\mathcal N}_{\Oe,r}(v,\kappa/2)$.
\end{remark}

\subsection{Interpolation}
The crucial step in the proof of the Large Deviation statement is based on the possibility to
approximate with partition functions on cells of a triangulation given in terms of $L^r$-neighbourhoods of linearizations of a minimiser of the rate functional. An important tool that will eventually allow to impose a boundary condition on each cell of the triangulation consists in switching between the corresponding partition
function $Z_{\Oe}({\mathcal N}_{\Oe,r}(v,\kappa))$ and the version $Z_{\Oe}({\mathcal N}_{\Oe,r}(v,2\kappa)\cap{\mathcal N}_{\Oe,R_0,\infty}(Z))$  with  an additional soft clamp $\abs{X(i)- Z(i)}<1$ enforced  in the boundary strip of the width $R_0> \diam(A)$ with $Z\in  {\mathcal N}_{\Oe,r}(v,\kappa)$ arbitrarily chosen.

Fixing parameters $\eta>0$ and  $N\in\N$, we will  slice the strip $\partial_{\eta/\e} \Oe$
into strips of width $\frac{\eta}{\e N}$ that will provide a framework for the interpolation.
Recalling   the notation
 \begin{equation}
F_{\kappa,\e}(v)=- \e^d \abs{\O}^{-1} \log Z_{\Oe}({\mathcal N}_{\Oe,r}(v,\kappa)), 
\end{equation}
we have the following claim.

\begin{lemma}[Interpolation]\hfill
\label{L:inter}  

\noindent
 Suppose that $U$  satisfies the assumptions (A1) and (A2). There exist constants $\kappa_0, b$, and ${\mathcal C}$
 (depending on $c, C$, $R_0$, $|\Omega|$, $|\partial\Omega|$, $p$,  and $m$)  and  a function $\e_0(\kappa, \eta, N)$ 
such that 
\begin{multline}
\label{E:Zfr<Zb}
 Z_{\Oe}({\mathcal N}_{\Oe,r}(v,\kappa))\le  Z_{\Oe}({\mathcal N}_{\Oe,r}(v,2\kappa)\cap{\mathcal N}_{\Oe,R_0,\infty}(Z))\times\\
 \times\exp\bigl\{\mathcal C \bigl((\tfrac{b+F_{\kappa,\e}(v)}{ N}+\eta + (\tfrac{N\kappa}\eta)^r) \e^{-d}+ \sum_{\substack{j\in S_{\eta/\e} \\ \tau_j(A)\subset\Oe}}  U({Z}_{\tau_j(A)})\bigr) \bigr\}
\end{multline}
 for any  $v\in L^{r}(\O)$, any $\kappa\le \kappa_0$, $\eta>0$, $N\in \N$, $Z\in  {\mathcal N}_{\Oe,r}(v,\kappa)$, and any $\e\le  \e_0(\kappa, \eta, N)$.
\end{lemma}
 \begin{remark}
Applying the lemma, we are only interested in the case when $r\ge p>1$. However, it is actually valid for any $r\ge p>0$.
 \end{remark}



\subsection{Equivalent definitions of the free energy}
Before attending to the proofs of our main Theorems, we will introduce several alternative partition functions yielding the  same free energy  $W(L)$ as that defined in Proposition~\ref{T:WL}.

As suggested above, one possibility is to relax the  boundary condition and to consider, instead, the configurations that are $\ell^r$-close
to $L$ by taking $Z_{\Oe}({\mathcal N}_{\Oe,r}(L,\kappa))$ as defined in \eqref{Zr,kappa}.
The same limit is obtained also by combining  both and considering the partition function
$Z_{\Oe}({\mathcal N}_{\Oe,r}(L,\kappa)\cap  {\mathcal N}_{\Oe,R_0, \infty}(L))$.

\begin{lemma}\hfill
\label{T:tildWL}

\noindent
 Suppose that $U$  satisfies the assumptions (A1) and (A2) with $r\ge p> 1$, $\tfrac1r>\tfrac1p-\tfrac1d$, and
let $L:\R^d\to\R^m$ be affine and $W(L)$ be as defined in \eqref{E:W}.
Then:

\noindent
a) We have 
\begin{equation}
-\lim_{\e\to 0}\e^d \abs{\O}^{-1} \log Z_{\Oe}({\mathcal N}_{\Oe,r}(L,\kappa)\cap  {\mathcal N}_{\Oe,R_0, \infty}(L))=W(L).
\end{equation}
  In particular, the limit  does not depend on $\kappa$ and  $\O$.

\noindent
b) Using 
\begin{equation}
W_{\kappa}(L) =-\limsup_{\e\to 0}\e^d \abs{\O}^{-1} \log Z_{\Oe}({\mathcal N}_{\Oe,r}(L,\kappa)),
\end{equation}
 we have  $\lim_{\kappa\to 0}  W_{\kappa}(L)=  W(L)$.

\noindent
c)
The free energy   ${W}(L)$  satisfies the bounds   $b\le {W}(L)\le B(L)$
with   $b=\frac{m}{p} \log c - \log  c(p,m) $,
where  $c(p,m)=\int_{\R^m} \exp\bigl(-\abs{\xi}^p\bigr)d\xi $, and  $B(L)=  C \bigl[\bigl(1+d\norm{L}^r \bigr) R_0^d+1+C\bigr].$
\end{lemma}
\begin{proof}

\noindent
a)
Using the shorthand
 $\widetilde Z_{\Lambda,\kappa}(L)=Z_{\Lambda}({\mathcal N}_{\L,r}(L,\kappa)\cap  {\mathcal N}_{\L,R_0, \infty}(L))$,
the existence of the limit and its independence on $\O$
follows easily by an obvious monotonicity in $\kappa$ and by
 standard methods with the help of approximative subadditivity (of $-\log \widetilde Z_{\Lambda,\kappa}(L)$) similar to \eqref{E:subaditivity},
\begin{equation}
\label{E:tildesubaditivity}
\log \tilde Z_{\L,\kappa} (L) \ge \log \tilde Z_{\L_1,\kappa}( L)+ \log \tilde Z_{\L_2,\kappa}( L)-B(L) |S(\L_1,\L_2)|.
\end{equation}
In addition to inserting the indicator $\1_{\L_1,L}(X)\1_{\L_2,L}(X)$,
 we also observe  the following inclusion,  ${\mathcal N}_{\L_1,r}(L,\kappa)\cap{\mathcal N}_{\L_2,r}(L,\kappa)\subset{\mathcal N}_{\L,r}(L,\kappa)$. Here, by  ${\mathcal N}_{\L_1,r}(L,\kappa)\cap{\mathcal N}_{\L_2,r}(L,\kappa)$ we mean a shorthand for the set
\begin{equation}
\{X\in (\R^m)^{\L}: X_{\L_1}\in
{\mathcal N}_{\L_1,r}(L,\kappa) \text{ and } X_{\L_2}\in
{\mathcal N}_{\L_2,r}(L,\kappa)\}.
\end{equation}
Indeed, for any $X\in {\mathcal N}_{\L_1,r}(L,\kappa)\cap{\mathcal N}_{\L_2,r}(L,\kappa)$, 
we have
\begin{equation}
\label{E:L12}
\sum_{i\in\L}\abs{X(i)-L(i)}^r\le  \kappa^r(\abs{\L_1}^{1+\frac{r}{d}}+\abs{\L_2}^{1+\frac{r}{d}})\le \kappa^r  \abs{\L}^{1+\frac{r}{d}}
\end{equation}
since, using $\xi$ to denote $\xi=  \frac{\abs{\L_1}}{\abs{\L}}$ with $1-\xi=  \frac{\abs{\L_2}}{\abs{\L}}$, we have $\xi^{1+\frac{r}{d}}+(1-\xi)^{1+\frac{r}{d}}\le 1$. 
 
To prove that the limit, denoted momentarily as 
\begin{equation}
\label{E:tilW}
\widetilde W_\kappa(L) =-\lim_{\e\to 0}\e^d \abs{\O}^{-1} \log Z_{\Oe}({\mathcal N}_{\Oe,r}(L,\kappa)\cap  {\mathcal N}_{\Oe,R_0, \infty}(L)),
\end{equation}
actually does not depend on $\kappa$, we first use the independence on $\O$ and 
consider the limit above with a cube $\O$. Notice that the cube $\overline{\O}_{\e}$ obtained as the cube $\O_{\e/2}$ rescaled by the factor 2 consists of a disjoint union of $2^d$ shifts of copies of the cube $\Oe$,
$\overline{\O}_{\e}=\cup_{k=1,\dots, 2^d} \tau_{i_k} (\Oe)$,  $\abs{\overline{\O}_{\e}}=2^d   \abs{\Oe}$.
 We have    $  \cap_{k=1,\dots, 2^d} {\mathcal N}_{ \tau_{i_k} (\Oe),r}(L,\kappa)\subset {\mathcal N}_{\overline{\O}_{\e},r}(L,\kappa/2)$.
 Indeed,  for any $X\in \cap_{k=1,\dots, 2^d} {\mathcal N}_{ \tau_{i_k} (\Oe),r}(L,\kappa)$, similarly as in \eqref{E:L12}, we have
 \begin{equation}
\sum_{i\in\overline{\O}_{\e}}\abs{X(i)-L(i)}^r\le  2^d \kappa^r   \abs{\Oe}^{1+\frac{r}{d}}=  2^d {\kappa}^r(2^{-d}\abs{\overline{\O}_{\e}})^{1+\frac{r}{d}}=(\tfrac{\kappa}2)^r \abs{\overline{\O}_{\e}}^{1+\frac{r}{d}}.
\end{equation}
As a result,
 \begin{equation}
\log  \tilde Z_{\overline{\O}_{\e},\kappa/2}(L)= \log \tilde Z_{\O_{\e/2},\kappa/2}(L) \ge   2^d   \log \tilde Z_{\Oe,\kappa}-B(L) 2^d \abs{\p\O} \e^{-d+1} R_0.
\end{equation}
Multiplying by $-(\e/2)^d \abs{\O}^{-1}$ and taking the limit $\e\to 0$, we get $ \widetilde W_{\kappa/2}(L)\le \widetilde W_{\kappa}(L)$.
On the other hand, $ \widetilde W_{\kappa/2}(L)\ge \widetilde W_{\kappa}(L)$ since $ \widetilde W_{\kappa}(L)$ is clearly decreasing in $\kappa$.

Combining discrete Poincar\'e inequality with the assumption (A1), we see that for any fixed $K$ and $\kappa$,
we have  ${\mathcal N}_{\Oe,r}(L,\kappa)^{\text{c}}\subset \mathcal M(K)$  for sufficiently small $\e$.
Then, by exponential tightness, for any fixed $\delta$ and sufficiently small $\e$,
\begin{equation}
Z_{\Oe}({\mathcal N}_{\Oe,r}(L,\kappa)\cap  {\mathcal N}_{\Oe,R_0, \infty}(L))\ge (1-\delta)Z_{\Oe}(  {\mathcal N}_{\Oe,R_0, \infty}(L)),
\end{equation}
implying that the limiting value $\widetilde W_\kappa=\widetilde W$ satisfies, for any $\delta$, the inequalities
\begin{multline}
W(L)\le \widetilde W(L) =-\lim_{\e\to 0}\e^d \abs{\O}^{-1}\log Z_{\Oe}({\mathcal N}_{\Oe,r}(L,\kappa)\cap  {\mathcal N}_{\Oe,R_0, \infty}(L))\le\\
\le -\lim_{\e\to 0}\e^d \abs{\O}^{-1}\log [(1-\delta)Z_{\Oe}(  {\mathcal N}_{\Oe,R_0, \infty}(L))]=W(L).
\end{multline}

 \noindent 
 b) Choosing $v=L$ and $Z=L$ ($Z(i)=L(i)$ for each $i\in\Oe$) in the interpolation lemma,  we get
 \begin{equation}
\widetilde W(L)\ge  W_{\kappa}(L)\ge \widetilde W(L)-
 \mathcal C \bigl(\tfrac{b+W_{\kappa}(L)}N+\eta + (\tfrac{N\kappa}\eta)^r + \eta \norm{L}^r C_{\p} \abs{\p \O}\bigr)
 \end{equation}
yielding 
 \begin{equation}
\widetilde W(L)\ge  \lim_{\kappa\to 0} W_{\kappa}(L)\ge \widetilde W(L)-
 \mathcal C \bigl(\tfrac{b+ \lim_{\kappa\to 0} W_{\kappa}(L)}N+\eta(1+ \norm{L}^r C_{\p} \abs{\p \O})  \bigr)
 \end{equation}
 for arbitrarily small $\eta$ and arbitrarily large $N$.

 \noindent
 c)
The lower bound follows from the inequality 
\begin{equation}
\widetilde Z_{\Oe,\kappa}(L)\le \o(m)  \bigl(c^{-m/p} c(p,m)\bigr)^{\abs{\Oe}},
\end{equation}
 obtained with help of (A1) and the bound from technical Lemma~\ref{L:tech} a) proven in Appendix.
For the upper bound, we just take into account that ${\mathcal N}_{\Oe,\infty}(Z) \subset {\mathcal N}_{\Oe,r}(v,\kappa)$
for sufficiently small $\e$ (cf.  Remark~\ref{R:tight}), to get 
\begin{equation}
\widetilde Z_{\Oe,\kappa}(L)\ge  
 \omega(m)^{\abs{\Oe}}\exp\bigl(-C \bigl(C(1+d R_0^d\norm{L}^r)+(1+R_0^d)\bigr)\abs{\Oe}\bigr)
\end{equation}
similarly as in \eqref{E:ZOe>} with the bound $U(L)\le C(1+d R_0^d \norm{L}^r)$ resulting  from (A2) (in the form from Remark~\ref{R:gradiendbound}). 
\end{proof}

Finally, we can enforce a version of approximate periodic boundary conditions yielding again the same free energy $W(L)$. Namely, 
consider the sets 
\begin{equation}
{\mathcal N}^{\per,\e}(L)=\{X:\Z^d\to \R^m\colon \abs{X(i+\tfrac{\bold{e}_j}{\e})-X(i)-\tfrac{L(\bold{e}_j)}{\e}}\le 2\  \forall i\in \Z^d, j=1,\dots,d \},
\end{equation}
and
\begin{equation}
{\mathcal N}^{\per,\e}_{r}(L,\kappa)=\{X\in {\mathcal N}^{\per,\e}(L)\colon  \norm{\Pe(X)-L}_{L^{r}([0,1]^d)} \le\kappa\},
\end{equation}
and define
\begin{equation}
Z_{[0,1]^d_{\e,R_0}}({\mathcal N}^{\per,\e}(L))= \int_{{\mathcal N}^{\per,\e}(L)}\exp\{-H_{[0,1]^d_{\e,R_0}} (X)\} \prod_{i\in [0,1]^d_{\e,R_0}}dX(i)
\end{equation}
and, similarly, also
$Z_{[0,1]^d_{\e,R_0}}({\mathcal N}^{\per,\e}_{r}(L,\kappa))$.
Here, we use $[0,1]^d_{\e,R_0}$ to denote the set 
\begin{equation}
\{i\in \Z^d\colon i_j\in[-R_0, \e^{-1} +R_0], j=1,\dots,d\}.
\end{equation}
Observing that
\begin{multline}
Z_{[0,1]^d_{\e,R_0}}({\mathcal N}_{[0,1]^d_{\e,R_0},r}(L,\kappa)\cap  {\mathcal N}_{[0,1]^d_{\e,R_0},R_0, \infty}(L))\le
Z_{[0,1]^d_{\e,R_0}}({\mathcal N}^{\per,\e}_{r}(L,\kappa))\le\\
\le Z_{[0,1]^d_{\e,R_0}}({\mathcal N}_{[0,1]^d_{\e,R_0},r}(L,\kappa))
\end{multline}
and applying the preceding lemma, we get
\begin{equation}
W(L)=-\lim_{\kappa\to 0}\lim_{\e\to 0} \e^d  \log Z^{\per}_{[0,1]^d_{\e,R_0}}({\mathcal N}^{\per,\e}_{r}(L,\kappa)).
\end{equation}
Similarly as in Lemma~\ref{T:tildWL} (b), we obtain the same limit also with $Z_{[0,1]^d_{\e,R_0}}({\mathcal N}^{\per,\e}(L))$:


\begin{lemma}\hfill
\label{T:WLalt}

\noindent
Suppose that (A1) and (A2) hold with $r\ge p>1$ and $\frac1r> \frac1p-\frac1d$. 
Then the free energy $W(L)$ from Proposition~\ref{T:WL} equals 
\begin{equation}
W(L)=-\lim_{\kappa\to 0}\lim_{\e\to 0} \e^d  \log Z^{\per}_{[0,1]^d_{\e,R_0}}({\mathcal N}^{\per,\e}_{r}(L,\kappa))=-\lim_{\e\to 0} \e^d  \log Z^{\per}_{[0,1]^d_{\e,R_0}}({\mathcal N}^{\per,\e}(L)).
\end{equation}
\end{lemma}

\subsection{Proof of Large Deviation Principle}

To prove  Theorem~\ref{T:LDab}, we begin by considerng, for any
$\r>0$ and $z \in Q(\r)=[-\tfrac{\r}2,\tfrac{\r}2]^d$,  the lattice
\begin{equation} 
\cL_{\r,z}=  (\r \Z)^d+z= \bigl\{x\in (\r \Z)^d+z\bigr\}.
\end{equation}
Our strategy will be  to approximate the integrals $Z_{\Oe}\bigl(O(v)\bigr)$ over  suitably chosen neighbourhoods
$O(v)$, $v\in W^{1,p}(\O)$, 
by a product of contributions over cubes obtained from $Q(\r)$ by shifts from $\cL_{\r,z}$.
Here $\r$ and $z$ will be chosen so that the function $v$ is, on each cube $x+Q(\r)$ for which $x+Q(\r)\subset \O$,  well approximated by its linear part  ${\mathcal L}_x v$ defined at $x$ by
${\mathcal L}_x v(y)= \nabla v(x)\cdot y+\dashint_{x+Q(\rho)}v(t) d t$ and, in the same time, the sum of 
the contributions ${W}(\nabla v(x))$ over the linear patches is well represented by the integral 
$\int_{\O} {W}(\nabla v(x))dx$.

To show that such a choice (of $\r$ and $z$)  is possible, we will use the following ``blow up'' lemma (the Corollary below)
with a function $f(x)$  related to an approximation of ${W}(\nabla v(x))$ and  the functions $v_{x,\r}$
representing the difference $v-{\mathcal L}_x v$; explicitly, we define
\begin{equation}
v_{x,\r}(y)=\tfrac1\r (v(x+\r y)-   {\mathcal L}_x v(\r y) )=\frac{v(x+\r y)-\dashint_{x+Q(\rho)}v(t) d t}\r -\nabla v(x)\cdot y
\end{equation} 
for any   $x\in \cL_{\r,z}$ and any
 $y\in Q=Q(1)$. For $v\in W^{1,p}(\O)$, the function $v_{\cdot,\r}(\cdot)$ is considered as belonging to $L^p(\O, W^{1,p}(\O))$.

\begin{lemma}
\label{L:blowup}
Let $r\ge p >1$,  $\frac1r>\frac1p-\frac1d$,  and let $v\in W_0^{1,p}(\R^d)$.
Then there exists a function  $\o_v: \R^+\to \R^+$ such that $\lim_{\r\to 0}\o_v(\rho)=0$  and

\noindent
a) $\dashint_{Q(\r)} \sum_{x\in  \cL_{\r,z}}\r^d (\int_{Q(1)} \abs{\nabla v_{x,\r}(y)}^pdy)dz  \le \o_v(\rho)$, 

\noindent
b) $\dashint_{Q(\r)} \sum_{x\in  \cL_{\r,z}}\r^d \bigl(\int_{Q(1)} \abs{v_{x,\r}(y)}^rdy\bigr)^{p/r}dz  \le \o_v(\rho)$.
\end{lemma}
\begin{proof} 
a)
Notice first that  for any $\o>0$ we can choose  $\r$ sufficiently small, to get 
\begin{equation}
\int_{\R^d}\int_{Q} \abs{\nabla v_{x,\r}(y)}^p dy dx=\int_{\R^d}\int_{Q} \abs{\nabla v(x+\r y)-\nabla v(x)}^p dy dx<\o
\end{equation}
by Lebesgue differentiation theorem.
Rewriting the integral $\int_{\R^d}\int_{Q} \abs{\nabla v_{x,\r}(y)}^p dy dx$ in the form of the sum $\int_{Q(\r)} 
\sum_{x \in \cL_{\r,z}} \int_{Q} \abs{\nabla v_{x,\r}(y)}^p dydz$,
we get
\begin{equation}
\label{E:meansum}
\frac{1}{\r^d}\int_{Q(\r)}\Bigl(\r^d \sum_{x \in \cL_{\r,z}} \int_{Q} \abs{\nabla v_{x,\r}(y)}^p dy\Bigr)dz<\o.
\end{equation}
\noindent
b) Follows from a) by Sobolev imbedding.
\end{proof}

\begin{corollary}
\label{L:blowup}
Let  $v\in W_0^{1,p}(\R^d)$, $\d>0$, $f\in L_1(\R^d)$, and let 
 $\ell<\int_{\R^d}f(x)dx$.  
Then there exists a constant $\r_0=\r_0(v,f,\ell,\d)$ and for each $\rho\le \rho_0$ a point $z\in Q(\r)$ such that 
\begin{equation}
\label{E:tauI}
\sum_{x \in \cL_{\r,z}} \r^{d}\int_{Q} \abs{\nabla v_{x,\r}(y)}^p dy<\delta 
\ \ \text{ and } \  \
\r^d  \sum_{x \in \cL_{\r,z}}f(x) > \ell.
\end{equation}
\end{corollary}

\begin{proof}

Interpreting the integral in \eqref{E:meansum} as the mean over $Q(\r)$ of the function in the brackets  and
using ${\cM}_{\delta}\subset Q(\r)$ to denote the set of points for which the first inequality in \eqref{E:tauI} is not valid, 
\begin{equation}
{\cM}_{\delta}=\{z\in Q(\r)\mid    \r^d\sum_{x\in\cL_{\r,z}(\O)}\int_{Q}\abs{\nabla v_{x,\r}(y)}^pdy\ge \d\},
\end{equation}
we can apply  Markov's inequality to get 
\begin{equation}
\abs{\cM_{\d}}\le \frac{  \o_v(\rho)  \r^d}{\d}  .
\end{equation}

On the other hand, assuming without loss of generality that $f\le K\1_{Q(R)}$ for some (large) $K$ and $R$ and
denoting $F(z)=\r^d  \sum_{x \in \cL_{\r,z}}f(x)$, we have $F\le K R^d$ with
the mean over $Q(\r)$ satisfying  
\begin{equation}
\dashint F(z) dz=\int_{Q(\r)} \r^{-d}F(z) dz=  \int f(y)dy> \ell.
\end{equation}
Denoting 
\begin{equation}
{\overline{\cM}}_{\ell}=\{z\in Q(\r)\mid   F(z) \le \ell\},
\end{equation}
we get
\begin{equation}
\int f(y)dy=\dashint F(z) dz  \le \ell \r^{-d}\abs{{\overline{\cM}}_{\ell}}+
K R^d(1-\r^{-d}\abs{{\overline{\cM}}_{\ell}}).
\end{equation}
Hence,
\begin{equation}
1-\r^{-d}\abs{{\overline{\cM}}_{\ell}}\ge \frac{\int  f(y) dy-\ell} {K R^d-\ell}.
\end{equation}

A point $z$ satisfying simultaneously both bounds in \eqref{E:tauI} thus exists once 
$1-\r^{-d}\abs{{\overline{\cM}}_{\ell}}-\frac{  \o_v(\rho)}{\d}> \epsilon$ for a fixed $\epsilon$ and $\r$ small.
For this to hold, it is enough to choose $\o_v$ (and corresponding $\r$) sufficiently small.
\end{proof}

Theorem~\ref{T:LDab} a) follows directly from the following lemma.

\begin{lemma}
\label{L:a}
For every $\d,\kappa, M\in(0,\infty)$ and any $v\in W^{1,p}(\O)$ with  $r\ge p >1$,  $\frac1r>\frac1p-\frac1d$, there exists $\tilde\kappa$ such 
\begin{equation}
Z_{\Oe}({\mathcal N}_{\Oe,r}(v,\tilde\kappa))\le \exp\bigl\{\e^{-d}\bigl(-\int_{\O} (W_\kappa(\nabla v(x))\wedge M) dx+\d\bigr)\bigr\}
\end{equation} 
for sufficiently small $\e$.
\end{lemma}

\begin{remark}
Whenever $\int_{\O}   W(\nabla v(x))dx<\infty$, we infer 
by Lebesgue theorem that
\begin{equation}
Z_{\Oe}({\mathcal N}_{\Oe,r}(v,\tilde\kappa))\le \exp\bigl\{\e^{-d}\bigl(-\int_{\O}   W(\nabla v(x))dx+\d\bigr)\bigr\}.
\end{equation} 
If $\int_{\O}   W(\nabla v(x))dx=\infty$, we can show that for any $M$ there exists $\e(M)$ so that
\begin{equation}
Z_{\Oe}({\mathcal N}_{\Oe,r}(v,\tilde\kappa))\le \exp\bigl\{-\e^{-d}M\bigr\}
\end{equation}
for $\e<\e(M)$.
\end{remark}

\begin{proof}
Replacing $v$ by an extension to $W^{1,p}(\R^d)$ with compact support, we apply Lemma~\ref{L:blowup} with $f(x)=(W_\kappa(\nabla v(x))\wedge M)\1_{\O}(x)$.
 Thus,   for any constant $\tilde\d>0$ and any $\r_0$, there exists 
  $\r<\r_0$ and a point $z\in Q(\r)$ such that 
  \begin{align}
\label{E:sigma}
&\sum_{x \in \cL_{\r,z}} \r^{d}\int_{Q} \abs{\nabla v_{x,\r}(y)}^p dy<\tilde\delta \ \ \ 
\text{ and } \\
&\r^d  \sum_{x \in \cL_{\r,z}}(W_\kappa(\nabla v(x))\wedge M)\1_{\O}(x)  > \int_{\O}  (W_\kappa(\nabla v(z))\wedge M) dz -\tilde\delta.
 \end{align}
Now, let us consider the vector $\boldsymbol{\kappa}=\{\kappa_x,  x \in \cL_{\r,z}\}$ 
with $\kappa_x^p=\int_{Q} \abs{\nabla v_{x,\r}(y)}^p dy$,  
and  the neighbourhood
\begin{equation}
O_{\boldsymbol{\kappa}}(v)=\bigcap_{x \in \cL_{\r,z}, \tau_x(Q(\rho))\cap\O\neq\emptyset}  {\mathcal N}_{\tau_x(Q(\r))\cap\O,r}({\mathcal L}_x v,\kappa_x).
\end{equation}
Cf.   \eqref{E:NrdL} for the definition of ${\mathcal N}_{\L,r}(v,\kappa)$.

Using (A1), we have 
$H_{\Oe}(X)\ge \sum_{x\in \cL_{\r,z}(\O)}H_{\tau_{x}(Q(\r))_\e\cap\Oe}(X)$.
Thus 
\begin{equation}
\label{E:Z(O)}
Z_{\Oe}\bigl(O_{\boldsymbol{\kappa}}(v)\bigr)
\le \prod_{x\in \cL_{\r,z}}
Z_{\tau_{x}(Q(\r))_\e}({\mathcal N}_{\tau_{x}(Q(\r))_\e\cap \Oe,r}({\mathcal L}_x v,\kappa_x)).
\end{equation}
Above, we take $Z_{\tau_{x}(Q(\r))_\e}({\mathcal N}_{\tau_{x}(Q(\r))_\e\cap \Oe,r}({\mathcal L}_x v,\kappa_x))=1$
whenever $\tau_{x}(Q(\r))_\e\cap \Oe=\emptyset$.

Taking now $\limsup$ of the appropriately rescaled logarithm of \eqref{E:Z(O)},
we get
\begin{multline}
\label{E:logZ(O)}
\!\!\!\!\!\!\!\!\limsup_{\epsilon\to0}\epsilon^d\log Z_{\Oe}\bigl(O_{\boldsymbol{\kappa}}(v)\bigr)
\le  \sum_{x\in \cL_{\r,z}}
\limsup_{\epsilon\to0}  \epsilon^d\log Z_{\tau_{x}(Q(\r))_\e}({\mathcal N}_{\tau_{x}(Q(\r))_\e\cap\Oe,r}({\mathcal L}_x v,\kappa_x))=\\
\!\!\!\!\!\!\!\!=-\r^d  \sum_{x \in \cL_{\r,z}}W_{\kappa_x}(\nabla v(x))\1_{\O}(x)
\le -\r^d  \sum_{x \in \cL_{\r,z}}(W_{\kappa_x}(\nabla v(x))\wedge M)\1_{\O}(x)\le\\
\!\!\!\!\!\!\!\!\le -\r^d  \sum_{x \in \cL_{\r,z}}(W_{\kappa}(\nabla v(x))\wedge M)\1_{\O}(x)
-\r^d  \sum_{\substack{x \in \cL_{\r,z}\\ \kappa_x>\kappa}}(W_{\kappa_x}(\nabla v(x))\wedge M- W_{\kappa}(\nabla v(x))\wedge M)\1_{\O}(x).
\end{multline} 
The absolute value of each term in the last last sum can be bounded by $M+\abs{b}$ with $b$ the lower bound from 
 Lemma~\ref{T:tildWL}. In the same time,
 the number of terms $n_{\kappa}$ for which $\kappa_x>\kappa$ is, in view of the bound $\sum \r^d\kappa_x^p<\tilde\delta$,
bunded by $n_{\kappa}\le \r^{-d}\frac{\tilde\delta}{{\kappa}^p}$.
  
 In summary, observing that for sufficiently small $\tilde\kappa$ the set ${\mathcal N}_{\L,r}(v,\tilde\kappa)$ is contained in
the  intersection $O_{\boldsymbol{\kappa}}(v)$ of a finite number of open sets,
 we are getting, for sufficiently small $\e$,
 \begin{equation}   
 \label{E:upper} 
\epsilon^d\log Z_{\Oe}\bigl(O_{\boldsymbol{\kappa}}(v)\bigr)\le
-\int_{\O} W_\kappa(\nabla v(x))dx+ (M+\abs{b})\frac{\tilde\delta}{{\kappa}^p} +\tilde\delta
\end{equation}
obtaining the claim by choosing sufficiently small $\rho$ and $\tilde\d$.
\end{proof}

For the lower bound, Theorem~\ref{T:LDab} b), we have to use Interpolation Lemma again (more precisely, we use  Lemma~\ref{T:tildWL}(a) that is based on it).
\begin{lemma}
a) For every $\d,\kappa\in(0,\infty)$ and any $v\in W^{1,p}(\O)$, we have
\begin{equation}
Z_{\Oe}({\mathcal N}_{\L,r}(v,\kappa))\ge \exp\bigl\{\e^{-d}\bigl(-\int_{\O}  W_{\kappa/2}(\nabla w(x)) dx-\d\bigr)\bigr\}
\end{equation} 
if $w$ is a piecewise linear function such that $\norm{w-v}_r\le \frac{\kappa}2 \abs{\O}^{\frac1r+\frac1d}$
and  $\e$ is sufficiently small.

\noindent
b) For every $\d,\kappa\in(0,\infty)$ and any $v\in W^{1,r}(\O)$, we have
\begin{equation}
Z_{\Oe}({\mathcal N}_{\L,r}(v,\kappa))\ge \exp\bigl\{\e^{-d}\bigl(-\int_{\O}   W(\nabla v(x)) dx-\d\bigr)\bigr\}
\end{equation} 
for   sufficiently small $\e$.
\end{lemma}
\begin{proof}

a)
For the first claim we first observe that
\begin{equation}
\bigcap_{j} \bigl( {\mathcal N}_{T_j,r}(w,\kappa/2)\cap  {\mathcal N}_{T_j,R_0,\infty}(w)\bigr)\subset {\mathcal N}_{\O,r}(v,\kappa)
\end{equation}
with $\{T_j\}$ denoting a triangulation consistent with piecewise linearity of $w$. Using the bound
$U(X_{\tau_i(A)})\le C(1+ U(X_{w,\e})+R_0^d)$ whenever $\tau_i(A)$ is 
 reaching over the boundaries of the linear parts of $w$ and then applying 
 Lemma~\ref{T:tildWL}(a) to evaluate each term $Z_{T_j}({\mathcal N}_{T_j,r}(w,\kappa/2)\cap  {\mathcal N}_{T_j,R_0,\infty}(w))$, we get
 the sought bound with 
a constant proportional to $\e^{-d+1}$ which is smaller than $\delta$ for sufficiently small $\e$.

\noindent
b)
For the second claim, we first notice (see Lemma~\ref{T:tildWL}(c)) that
$W(L)\le C(1+\norm{L}^r)$. It follows that, if $w_n$ is a sequence of piecewise linear functions such that $\int \abs{\nabla(w_n(x)-v(x)}^r dx\to0$, then 
$W_{\kappa/2}(\nabla w_n(x))$ is equiintegrable.
Using  the bound
$ W_{\kappa/2}(L)\le  W_{\kappa/4}(\widetilde L)$ valid for $\norm{L-\widetilde L}<\kappa/4$,
we conclude that 
\begin{equation}
\lim\int W_{\kappa/2}(\nabla w_n(x))dx\le 
 \int W_{\kappa/4}(\nabla v(x))dx\le \int W(\nabla v(x))dx.
 \end{equation}
\end{proof}
 
\subsection{Proof of Proposition~\ref{P:QC}}
Consider $v\in W^{1,r}_0(\O)+L$.
According to Theorem~\ref{T:LDab} b), we have
$ \lim_{\kappa\to 0} \liminf_{\e\to 0} F_{\kappa,\e}(v)\le \frac1{\abs{\O}}\int_{\O} W(\nabla v(x)) dx$ with
$F_{\kappa,\e}(v)=- \e^d \abs{\O}^{-1} \log Z_{\Oe}({\mathcal N}_{\Oe,r}(v,\kappa))$.
Using the obvious inequality 
\begin{equation}
Z_{\Oe}({\mathcal N}_{\Oe,r}(v,2\kappa)\cap {\mathcal N}_{\Oe,R_0,\infty}(L))\le 
Z_{\Oe}( {\mathcal N}_{\Oe,R_0,\infty}(L))
\end{equation} and  Interpolation Lemma, we get
\begin{equation}
Z_{\Oe}({\mathcal N}_{\Oe,r}(v,\kappa))\le Z_{\Oe}({\mathcal N}_{\Oe,R_0,\infty}(L))
\exp\bigl\{ \e^{-d}\mathcal C \bigl(\tfrac{b+F_{\kappa,\e}(v)}N+\eta + (\tfrac{N\kappa}\eta)^r + \eta \norm{L}^r C_{\p} \abs{\p \O}\bigr)\bigr\}.
 \end{equation}
In view of Proposition~\ref{T:WL} thus
 \begin{equation}
 \lim_{\kappa\to 0} \liminf_{\e\to 0} F_{\kappa,\e}(v)= W(L),
 \end{equation}
 implying the claim
 \begin{equation}
 W(L)\le\frac1{\abs{\O}} \int_\O W(\nabla v(x))dx.
 \end{equation}
{}\hfill\hfill\qed

\subsection{Non-convexity of the free energy}
\label{S:nonconv}
Let us  briefly discuss the fact that the free energy $W(L)$ may be, in general, a  non-convex function of $L$ (Remark~\ref{R:non-convex}).
The idea hinges on the fact that an addition, to the original Hamiltonian $H^{(0)}$, of a term in the form of a hugely non-convex discrete null Lagrangian leads to a directly computable addition to the original free energy $W_0$ yielding a non-convex sum $W(L)$.
It suffices to assume that the free energy $W_0$ is bounded from above and below,  $ b\le W_0(L)\le B$,  for all $L$ with $\norm{L}\le 1$.
 An example might be the potential $U(X)= \abs{\nabla X(0)}^p$  for which $W_0(L)\sim \norm{L}^p$.

In more details, consider, for simplicity, the case $d=m=2$. Let $Q$ be a unit square $Q=(i_0,i_1,i_2,i_3)$ in $\Z^2$ (with $i_0=(0,0), i_1=(1,0), i_2=(0,1)$, and $i_3=(1,1)$) and, for any $X\in (\R^m)^{Q}$, let $V(X_Q)$ be defined by 
\begin{equation}
V(X_Q)=\tfrac12\det\bigl(X(i_1)-X(i_0), X(i_2)-X(i_0)\bigr)+\tfrac12\det\bigl(X(i_1)-X(i_3), X(i_2)-X(i_3)\bigr).
\end{equation}
Geometrically, $V(X_Q)$ yields the area of the rectangle $\bigl(X(i_0),X(i_1), X(i_3),X(i_2)\bigr)$.
In particular, for an affine map $L$, $V(L)$ is the area   of the deformed square $L(Q)$. Thus,
$V(\text{id})=1$ for the identity map $\text{id}$,  $\text{id}(i)=i$, and $V(L^{(0)})=0$ for the zero map $L^{(0)}$, $L^{(0)}(i)=0$.

Consider the Hamiltonian 
\begin{equation}
H_{\Oe}(X)= H^{(0)}_{\Oe}(X) +M \negthickspace\negthickspace\sum_{j\in\Z^d\colon \tau_j(Q)\subset \Oe} 
\negthickspace\negthickspace (1-V(X_{\tau_j(Q)}))=H^{(0)}_{\Oe}(X) +\begin{cases} 0\ \text{ for }\ L=\text{id},\\
M\ \text{ for }\ L=L^{(0)},
\end{cases}
\end{equation}
where  $H_{\Oe}^{(0)}(X)$ is the original Hamiltonian
and $M>0$ is a constant.  The crucial point is that the term $V$ is a discrete null Lagrangian (see e.g. \cite{FT}):
the value of the additional term
\begin{equation}
H^*_{\Oe}(X)= M \negthickspace\negthickspace\sum_{j\in\Z^d\colon \tau_j(Q)\subset \Oe} 
\negthickspace\negthickspace (1-V(X_{\tau_j(Q)}))
\end{equation}
 depends only on $X$  in the boundary layer, $H^*_{\Oe}(X)=H^*_{\Oe}(\overline X)$ if $X\vert_{\p_{R_0}\Oe}=\overline X\vert_{\p_{R_0}\Oe}$. More precisely, $H^*_{\Oe}(X)$ equals 
 $M(\text{vol}(\text{id}) - \text{vol}(X))$, where $\text{vol}(X)$ is
 the signed volume of the envelope of the set points $X(i), i\in\Oe$. 
 
 We have

\begin{lemma}
Let  $U$ be a 
potential whose corresponding free energy $W_0$  is bounded from above and below,
$W_0(L) \in (b,B)$,  for every $L$ such that $\norm{L}\le 1$. Then the free energy $W$ corresponding to the Hamiltonian 
$H^{(0)}+H^{*}$ is non-convex for M sufficiently large.
\end{lemma}
\begin{proof}
Consider $L_1= \text{id}$ and $L_2=-\text{id}$. For any $X\in {\mathcal N}_{\Oe,R_0, \infty}(L_1)$, we have  $H^*_{\Oe}(X)= H^*_{\Oe}(L_1)+ O(\e^{d-1})\norm{L_1}= O(\e^{d-1})$ since the volume spanned by $X(i), i\in\Oe$ differs from the volume spanned by $\Oe$ at most by $O(\e^{d-1})\norm{L_1}$. Similarly $H^*_{\Oe}(X)=O(\e^{d-1})$ for  any $X\in {\mathcal N}_{\Oe,R_0, \infty}(L_2)$.
Thus 
$W(L_1)=W^{(0)}(L_1)$ and $W(L_2)=W^{(0)}(L_2)$. Given that   $\tfrac12 L_1 +\tfrac12 L_2=L^{(0)}$ and 
$H^*_{\Oe}(X)=M\abs{\Oe} + O(\e)$ for every $X\in {\mathcal N}_{\Oe,R_0, \infty}(L_0)$,
we get $W(L^{(0)})=W^{(0)}(L^{(0)}) +M\ge b+M> B\ge \tfrac12 W^{(0)}(L_1)+ \tfrac12 W^{(0)}(L_2)=\tfrac12 W(L_1)+ \tfrac12 W(L_2)$
once $M>B-b$.
\end{proof}

\subsection{Proof of Theorem~\ref{T:GYG}}

We will  use a particular case of the following Lemma  formulated in an abstract setting.
It is based on the following two standard facts.

\noindent
(1) Let $\cX$ be a topological space,
$\cK_\ell\subset\subset \cX$ a sequence of its compact separable subspaces, and $\e_\ell\to 0$ a sequence of positive numbers. Then the set of Borel probability measures with uniform tightness condition,
\begin{equation}
\cM_{(\e_\ell)}=\{\alpha\in BC(\cX)^\ast\colon 0\le \alpha, \alpha(1)=1, 
\alpha(\cX\setminus\cK_\ell)\le \e_\ell\}
\end{equation}
is weakly compact.
Here, as usually, $\alpha(\cX\setminus\cK_\ell)=\sup\{\alpha(\varphi)\colon \varphi\in C(\cX), \varphi\le \1_{\cX\setminus\cK_\ell}\}$.

Moreover, if we have a sequence $\mu_n$ of Borel probability mesures on $\cX$ such that
$\mu_n(\cX\setminus\cK_\ell)\le \e_\ell$ for all $n> n(\ell)$, then a subsequence converges weakly to
a Borel probability measure $\mu\in  \cM_{(\e_\ell)}$.

\noindent
(2)
Let $\cX_1$ and  $\cX_2$ be topological spaces and $\cK_{1,\ell}\subset\subset\cX_1$ and 
$\cK_{2,\ell}\subset\subset\cX_2$ be sequences of separable compact 
subspaces and $\mu\in BC(\cX_1\times\cX_2)^\ast$ be such that 
$\lim_{\ell\to\infty}\mu(\cX_1\times\cX_2\setminus \cK_{1,\ell}\times\cK_{2,\ell})=0$.
Then there exists a weakly measurable map $\nu: \cX_1\to BC(\cX_2)^\ast$ so that
 \begin{equation}
\mu(\varphi)=\int_{\cX_1}\nu_x(\varphi(x,\cdot)) d\mu_1,
\end{equation}
where $\mu_1$ is the marginal of $\mu$. Moreover, $\nu_x(1)=\lim \nu_x(\cK_{2,\ell})$
for almost all $x$.

Now, let $(S_n, \mu_n, \Sigma_n)$ be a sequence of probabilities,
$\cX_1$, $\cX_2$, and $\cX_3$ be topological spaces and $\cK_{k,\ell}\subset\subset\cX_k$, $k=1,2,3$,
be sequences of separable compact subspaces, and let  $\lambda$ be a Borel probability measure on $\cX_3$. 
Further, let a sequence of  mappings
$\cT_n: S_n\times \cX_3\to \cX_1\times\cX_2$ be given that are measurable with respect to the Borel $\sigma$-algebras on $\cX_1$, $\cX_2$. 

We say that a sequence $\cT_n$ fulfills a \emph{uniform tightness condition} (with respect to probabilities $\mu_n\times \lambda$ and a sequence ($\e_\ell$), $\e_\ell\to 0$) if
\begin{equation}
\lambda(\cX_3\setminus \cK_{3,\ell})+(\mu_n\times\lambda)(\cT_n^{-1}((\cX_1\times \cX_2)\setminus (\cK_{1,\ell}\times \cK_{2,\ell}))\le
\e_\ell
\end{equation}
for every $\ell$ and $n\ge n(\ell)$.

In this setting, the observations (1) and (2) lead to the following claim.

\begin{lemma}
\label{L:abstr}
Given a sequence $\cT_n$ fulfilling a uniform tightness condition,
there exists a subsequence $n_k\to\infty$, Borel probability measures $\gamma$ on $\cX_1$ 
and  $\l$ on $\cX_3$ such that 
\begin{equation}
\gamma(\cX_1\setminus \cK_{1,\ell})\le \e_\ell \text{ and }
\lambda(\cX_3\setminus \cK_{3,\ell})\le \e_\ell \text{ for all }\ell,
\end{equation}
and a mapping $\nu:\cX_1\times\cX_2\to \cP(\cX_2)$ that is weakly measurable (with respect to the weak topology of $\bigl(BC(\cX_2)\bigr)^\ast$) satisfying  
\begin{equation}
0\le \nu_{x_1,x_3}, \nu_{x_1,x_3}(1)=1\ \text{ and }\
\int_{\cX_1\times\cX_3} \nu_{x_1,x_3}(\cX_2\setminus\cK_{2,\ell}) d\gamma(x_1) d\lambda(x_3)\le\e_\ell
\end{equation}
 for almost all $x_1$ and $x_3$,
such that
\begin{multline}
\lim_{k\to\infty}\int \varphi(\cT_{n_k}(\o,x_3),x_3) d\mu_{n_k}(\o) d\lambda(x_3) =\\=
\int_{\cX_1\times\cX_3} \bigl[\textstyle{\int}_{\cX_2}\varphi(x_1,x_2,x_3)d\nu_{x_1,x_3}(x_2)\bigr] d\gamma(x_1) d\lambda(x_3).
\end{multline}
for any bounded and continuous test function $\varphi$ on $\cX_1\times \cX_2\times \cX_3$.
\end{lemma}

We will apply the above Lemma in the following situation.

We take $\cX_1=W_0^{1,p}(\O)$,  $\cX_2=BC(S)^\ast$, and $\cX_3=\O$
(recall that $S=(\R^m)^{\Z^d}/\R^m$).
Further, we consider the  sets 
\begin{equation}
\cK_{1,\ell}=\{ v\in W_0^{1,p}(\O)+u\colon \norm{\nabla v}_{p}\le \ell\}.
\end{equation}
Note that by the Poincar\'e inequality, $\cK_{1,\ell}$ is bounded in the norm topology of $W_{1,p}(\O)$.
To define $\cK_{2,\ell}$, we first introduce the sets 
\begin{equation}
B_{\Lambda,\ell}=\{X\in S\colon \sum_{i\in\Lambda}\abs{\nabla X(i)}^p\le \ell\}
\end{equation}
and 
\begin{equation}
\widetilde{\cK}_{2,\ell}=\{\mu\in BC(S)^\ast\colon 
0\le \mu, \abs{\mu}\le 1, \supp \mu\subset B_{\Lambda_N,2^N \abs{\Lambda_N} \ell} \text{ for every } N\}.
\end{equation}
Here, $(\Lambda_N)$ is the sequence of sets $\Lambda_N=[-N,N]^d\cap\Z^d$.
Then 
\begin{multline}
{\cK}_{2,\ell}=\{ \mu\in BC(S)^\ast\colon 0\le \mu, \abs{\mu}\le 1, 
\text{ for any } k \\
\text{ there exists } \mu_k\in \widetilde{\cK}_{2,2^k \ell}
 \text{ so that } \abs{\mu-\mu_k}\le 2^{-k/2}\}.
\end{multline}
Clearly, the sets ${\cK}_{1,\ell}$, $\widetilde{\cK}_{2,\ell}$, and ${\cK}_{2,\ell}$ are compact separable in the weak topology.
Also, we take  $\lambda$, the normalized Lebesgue measure on $\O$, and  
\begin{equation}
\cK_{3,\ell}=\O\setminus\p_{1/\ell}\O=\{x\in\O:\dist(x,\p\O)\ge 1/\ell\},
\end{equation}
  and  for the probabilities $(S_n, \mu_n, \Sigma_n)$ we take $S_n=S$ and $\mu_n=\mu_{\e_n,u}$
with $\e_n\to 0$ as $n\to\infty$.

Finally, we introduce the family of mappings 
$\cT_{\Lambda,n}: S\times \O\to (W_0^{1,p}(\O)+u)\times BC(S)^\ast$
defined by 
\begin{equation}
\cT_{\Lambda,n}(X,x)=(\Pie(X), \mu_\Lambda(\cdot\mid \tau_{\lfloor x/\e_n\rfloor}(X)).
\end{equation}
We will consider the sequences $\cT_{\Lambda_N,n}$, first in $n$ and then in $N$,
and show that they satisfy a uniform tightness condition.
To this end we verify the following bounds.

\begin{lemma} There exist fixed constants $\tilde c$ and  $\tilde\ell$ such that, uniformly in $\Lambda$,
\label{L:mueu}

\noindent
(a) $\mu_{\e,u}\bigl(\{X\in S\colon \int_{\O}\abs{\nabla\Pi_{\e}(X)(x)}^p d\l(x)\ge\ell \}\bigr)\le \exp(-\tilde c\ell\e^{-d})$ and

$\int (\int_\O (\abs{\nabla\Pi_{\e}(X)(x)}^p-\ell)_+d\l(x))d\mu_{\e,u}(X)\le \exp(-\tilde c\ell\e^{-d})$  for any $\ell\ge \tilde\ell$.

\noindent
(b) $\int_{\cK_{3,\ell}}\int (1-\1_{\cK_{2,\ell}})(\mu_{\Lambda}(\cdot \mid \tau_{\lfloor x/\e\rfloor}(X)) 
d\mu_{\e,u}(X) d\lambda(x)\le \tilde c(1+\tfrac{\norm{u}_{1,p}}{\tilde\ell})\tfrac{\tilde\ell}\ell$ whenever 

\noindent
\ \ \ \ $\e\,\diam\Lambda<1/\ell$.
\end{lemma}
\begin{proof}\

\noindent
(a) 
is an immediate consequence  of the assumption (A1) and Exponential Tightness once we observe that $\{X\colon \sum_{i\in \Oe} \abs{\nabla X(i)}^p > \e^{-d}\ell\}\subset \mathcal M_K$ with $K=\frac{c\ell}{\abs{\O}}$.

\noindent
(b) 
Notice that $\lambda(\p_{1/\ell}\O)\le   C_{\p}  \frac{\abs{\p\O}}{\abs{\O}}\frac1\ell$
in view of the condition (A${}_{\boldsymbol{\partial}}$).
On several occasions we will use the DLR condition in the following form:
whenever $\L\subset\Oe$ and $f,g$ are  measurable cylinder functions  on $S$ with $g$ living on   $\Z^d\setminus\L$, then
\begin{equation}
\label{E:DLR}
\int \mu_{\L}(f\mid X)g(X)d \mu_{\e,u}(X)=\mu_{\e,u}(fg).
\end{equation}
Using this (with $g=1$) and assuming that $\e\,\diam \,\L<1/\ell$, we have
\begin{multline}
\label{E:DLRtau}
\int_{\cK_{3,\ell}}\int \mu_{\Lambda}(f\mid \tau_{\lfloor x/\e\rfloor}(X))
d\mu_{\e,u}(X) d\lambda(x)=\\=
\int_{\cK_{3,\ell}}\int \mu_{\tau_{\lfloor x/\e\rfloor}(\Lambda)}(f\circ \tau_{\lfloor x/\e\rfloor}\mid X)
d\mu_{\e,u}(X) d\lambda(x) =
\int_{\cK_{3,\ell}}\mu_{\e,u}(f\circ \tau_{\lfloor x/\e\rfloor})d\lambda(x). 
\end{multline}
Taking $f=\1_{{}{B_{\Lambda_N,2^{N+k} \abs{\Lambda_N}\ell}^{\text{c}}}}$, we get
\begin{multline}
\int_{\cK_{3,\ell}} \int (1-\1_{\cK_{2,\ell}})(\mu_{\Lambda}(\cdot \mid \tau_{\lfloor x/\e\rfloor}(X)) 
d\mu_{\e,u}(X) d\lambda(x)\le\\
\le\sum_{k} 2^{k/2}\! \int_{\cK_{3,\ell}} \mu_{\e,u}(\cup_N B_{ \tau_{\lfloor x/\e\rfloor}(\Lambda_N),2^{N+k} \abs{\Lambda_N} \ell}^{\text{c}})d\lambda(x)\le\\
\le \sum_{k,N} 2^{k/2} \int_\O\int  \tfrac1{2^{N+k}\abs{\Lambda_n}\ell} \sum_{i\in \tau_{\lfloor x/\e\rfloor}(\Lambda_n)}\abs{\nabla X(i)}^p d\mu_{\e,u}(X) d\lambda(x)
\le\\ \le
\sum_{k,N} 2^{-N-k/2}\!\! \int  \!\!\tfrac{\e^d}{\ell}  \sum_{i\in \Z^d}\abs{\nabla X(i)}^p d\mu_{\e,u}(X) \le\\
\le
\sum_{k,N} 2^{-N-k/2}  \tfrac{1}{\ell} \bigl[\norm{u}_{1,p}+ \tilde\ell + \sum_{\ell\ge0}2^{\ell+1} \tilde\ell\mu_{\e,u}(\{X\colon \e^d \sum_{i\in\Oe}  \abs{\nabla X(i)}^p \in[2^{\ell} \tilde\ell, 2^{\ell+1} \tilde\ell\}] \bigr]\le
\\ \le
\sum_{k,N} 2^{-N-k/2}  \tfrac{\tilde\ell}{\ell} \bigl[\tfrac{\norm{u}_{1,p}}{\tilde\ell}+ 1 + \sum_{\ell\ge0}2^{\ell+1} 
\exp(-\tilde c2^\ell \tilde\ell\e^{-d}) \bigr]\le\\
\le \tfrac{\tilde\ell}\ell \sum_{k,~N}  2^{-N-k/2}(\tfrac{\norm{u}_{1,p}}{\tilde\ell}+1+2\sum_{\ell} 2^{-\ell}) \le 2\sqrt 2(\sqrt 2+1) (5+\tfrac{\norm{u}_{1,p}}{\tilde\ell})\tfrac{\tilde\ell}\ell
\end{multline}
once  $\tilde\ell$ is large enough (and for $\e$ small enough) so that $t^2 e^{-\tilde c \tilde\ell\e^{-d} t}\le 1$ for any $t\ge 0$.
\end{proof}
Applying now Lemma~\ref{L:abstr}, we get the claim \eqref{L:GYGdesint} for any $\varphi\in BC((W^{1,p}_0(\O)+u)\times BC(S)^\ast_+\times \O)$.
To extend \eqref{L:GYGdesint} to a more general class of test functions, we will use the following Lemma.
\begin{lemma}
Let $\nu_{\widetilde\Lambda,\e}(\varphi)=\mu_{\e,u}\bigl(\textstyle{\int_{\O}} \varphi(\Pie(X), \mu_{\widetilde\Lambda}(\cdot|\tau_{\lfloor{x}/{\e}\rfloor}(X)), x)d\l(x)\bigr)$. Consider the weak closure $\cM$ of the set
$\bigl\{\nu_{\widetilde\Lambda,\e}: \e\,\diam{\widetilde\Lambda}<\rho\bigr\}$.
 Further, let
\begin{equation}
\psi: (W^{1,p}_0(\O)+u)\times BC(S)^\ast_+\times
 \O\to \R
\end{equation}
be weakly continuous and such that for any $\delta$ it fulfills the growth condition
\begin{equation}
\label{E:PhiLippsi} 
\abs{\psi(v,\mu,x)}\le \eta(x)\exp\{ c_{\psi} \norm{\nabla v}_{p}^p)\}  \mu\bigl(\delta g_{\L} +C(\delta)\bigr)
\end{equation}
with  fixed $\Lambda$, $g_{\L}(X)=\textstyle{\sum}_{i\in\Lambda}\abs{\nabla X(i)}^p$, $\eta\in C_0(\O)$ such that $\dist(\supp \,\eta,\p\O)>\rho$, and a constant $c_{\psi}$ depending only on $\psi$.
Then, if for some $\nu_n, \nu \in \cM$ we have $\nu_n(\varphi)\to\nu(\varphi)$ for all
$\varphi\in BC((W^{1,p}_0(\O)+u)\times BC(S)^\ast_+\times \O)$, 
then also $\nu_n(\psi)\to\nu(\psi)$.
\end{lemma}
\begin{remark}
Notice that, due to Poincar\'e inequality, the set of test functions above is the same if we replace
$\norm{\nabla v}_{p}^p$ in \eqref{E:PhiLippsi}   by $\norm{v}_{1,p}^p$.
\end{remark}

\begin{proof}
Since
$\psi\min(1, \frac{\ell}{\abs{\psi}})\in BC((W^{1,p}_0(\O)+u)\times BC(S)^\ast_+\times \O)$,
 it is enough to prove that 
\begin{equation}
\nu((\psi-\ell)_+)\le \o(\ell) \text{ with } \lim_{\ell\to\infty}\o(\ell)=0
\end{equation}
uniformly in $\nu\in\cM$. For any $k>0$, let us decompose $\psi$ as follows,
\begin{equation}
\psi= \psi \1_{\norm{\nabla v}_{p}^p\le k}+\psi \1_{\norm{\nabla v}_{p}^p>k}.
\end{equation}

For the first term, we notice that
\begin{equation}
\abs{\psi(v,\mu,x)}\1_{\norm{\nabla v}_{p}^p\le k}(v)\le \eta(x) e^{c_{\psi} k} \bigl(C(\delta) +\delta  \mu (g_{\L})\bigr),
\end{equation}
implies
\begin{equation}
(\psi(v,\mu,x)-\ell)_+\1_{\norm{\nabla v}_{p}^p\le k}(v)\le \delta \eta(x) e^{c_{\psi} k}  \mu (g_{\L})
\end{equation}
once $\ell > \ell(k,\delta)=\norm{\eta}e^{c_{\psi} k} C(\delta)$ and thus,  for $\nu=\nu_{\widetilde\Lambda,\e}$
with $\dist(\supp \,\eta,\p\O)>\e\,\diam \widetilde\L$, we get
\begin{multline}
\nu\bigl((\psi-\ell)_+\1_{\norm{\nabla v}_{p}^p\le k}\bigr)\le \delta e^{c_{\psi} k}\int\bigl( \int\eta(x)   \mu_{\widetilde\Lambda}(g_{\L}|\tau_{\lfloor{x}/{\e}\rfloor}(X))\,d \l(x) \bigr)d\mu_{\e,u}(X)=\\
= \delta e^{c_{\psi} k} \int\eta(x)  \bigl(\int  g_{\tau_{\lfloor{x}/{\e}\rfloor}\L}(X)d\mu_{\e,u}(X)\bigr)\, d\l(x)
\le\delta e^{c_{\psi} k}\norm{\eta} \abs{\L} \widetilde C
\end{multline}
with $\widetilde C$ denoting the bound on $\int \norm{\nabla\Pie X}_{p}^p d\mu_{\e,u}(X)$
(cf. Lemma~\ref{L:mueu} (a)).
Here, we first  used that 
$\int   \mu_{\widetilde\Lambda}(g_{\L}|\tau_{\lfloor{x}/{\e}\rfloor}(X))d\mu_{\e,u}(X)=\mu_{\e,u}(g_{\L}\circ \tau_{\lfloor{x}/{\e}\rfloor})$
similarly as in  \eqref{E:DLRtau} and then bounded $
\int\eta(x)   g_{\tau_{\lfloor{x}/{\e}\rfloor}\L}(X)\, d\l(x)\le \norm{\eta} \abs{\L} \e^d \sum_{i\in\Oe}\abs{\nabla X(i)}^p$.

Further,  we will  slice  
$\psi \1_{\norm{\nabla v}_{p}^p>k}= \sum_{j\ge 0}\psi_{2^j k}$
with the functions $\psi_n= \psi \1_{  n<\norm{\nabla v}_{p}^p\le 2 n}$ satisfying the bound (with $\delta=1$),
\begin{equation}
\abs{\psi_n(v,\mu,x)}\le \eta(x) e^{2 c_{\psi} n} \bigl( \mu (g_{\L})+C(1) \bigr)\1_{ n<\norm{\nabla v}_{p}^p\le 2 n}(v).
\end{equation}
In preparation for the evaluation of $\nu\bigl(\psi_n)=\nu_{\widetilde\Lambda,\e}(\psi_n)$, we use 
$(\widetilde \L\cup \L)_1$ to denote the $1$-neighbourhood of $\widetilde \L\cup \L$ and for any
$X,Y\in (\R^m)^{\Z^d}$ we bound 
\begin{multline}
\label{E:trl}
g_{\L}(Y)
\1_{ n<\norm{\nabla\Pie(X)}_{p}^p\le 2n}(X)\le \\
\!\!\!\!\!\!\!\le (g_{\L}(Y)-g_{(\widetilde \L\cup \L)_1}(\tau_{\lfloor{x}/{\e}\rfloor}(X))_+\1_{ n<\norm{\nabla\Pie(X)}_{p}^p}(X)
+
g_{(\widetilde \L\cup \L)_1}(\tau_{\lfloor{x}/{\e}\rfloor}(X))
\1_{ n<\norm{\nabla\Pie(X)}_{p}^p\le 2n}(X) \le\\
\le
 g_{\L}(Y) \1_{ n/2<\norm{\nabla\Pie(X)}_{p}^p-\e^d g_{\tau_{\lfloor{x}/{\e}\rfloor}(\widetilde\L\cup\L)_1}(X)}(X)
+ (g_{\L}(Y)-\e^{-d} n/2)_+
+\\+
g_{(\widetilde \L\cup \L)_1}(\tau_{\lfloor{x}/{\e}\rfloor}(X))
\1_{ n<\norm{\nabla\Pie(X)}_{p}^p\le 2n}(X) .
\end{multline}
Notice that
$\norm{\nabla\Pie(X)}_{p}^p-\e^d g_{\tau_{\lfloor{x}/{\e}\rfloor}(\widetilde\L\cup\L)_1}(X)=
\e^d g_{\Oe\setminus \tau_{\lfloor{x}/{\e}\rfloor}(\widetilde\L\cup\L)_1}(X)$ and thus the right hand side above
actually does not depend on $X(i), i\in \tau_{\lfloor{x}/{\e}\rfloor}(\widetilde\L\cup\L)$.
As a result, using \eqref{E:DLR} we get
\begin{multline}
\int   \mu_{\widetilde\Lambda}(g_{\L}|\tau_{\lfloor{x}/{\e}\rfloor}(X))\1_{ n<\norm{\nabla\Pie(X)}_{p}^p\le 2n}d\mu_{\e,u}(X) 
\le 
\mu_{\e,u}\bigl((g_{\L}\circ \tau_{\lfloor{x}/{\e}\rfloor}) \1_{ n/2<\norm{\nabla\Pie(\cdot)}_{p}^p}\bigr)
+ \\+
\mu_{\e,u}\bigl((g_{\L}\circ \tau_{\lfloor{x}/{\e}\rfloor}-\e^{-d} n/2)_+\bigr)
+
\int  g_{(\widetilde \L\cup \L)_1}(\tau_{\lfloor{x}/{\e}\rfloor}(X))
\1_{ n<\norm{\nabla\Pie(X)}_{p}^p\le 2n}(X)d\mu_{\e,u}(X)
\end{multline}
Hence, once $\e\, \diam\widetilde\L < \dist(\supp \, \eta, \p\O)$, we bound $\nu\bigl(\psi_n)$, up to a prefactor $e^{2c_{\psi_n} n}\norm{\eta}$, by
\begin{equation}
\int\bigl( C(1)+  2n\abs{\widetilde\L\cup\L} 
\1_{n<\norm{\nabla\Pie(X)}_{p}^p}(X)+
2\abs{\L} 
\norm{\nabla\Pie(X)}_{p}^p\1_{\norm{\nabla\Pie(X)}_p^p+\ge n/2}
(X) \bigr)d\mu_{\e,u}(X).
 \end{equation}
Thus, for $n$ large,
 \begin{equation}
\nu\bigl(\psi_n)\le e^{2c_{\psi_n} n}\norm{\eta}
\bigl( (C(1)+  \abs{\widetilde\L\cup\L} 2n ) e^{-\tilde c n \e^{-d}}+
\bar C\abs{\L} 
 e^{-\tilde c n/2 \e^{-d}} \bigr)
\end{equation}
implying the claim.
\end{proof}

To show that the family $\nu_{x,v}$ of Young measures has support in the set of Gibbs measures,
observe that $\mu\in\mathcal G$ iff
\begin{equation}
\mu(\mu_{\Lambda}(f | \cdot))-\mu(f)= \varphi_{\L,f}(\mu)=0
\end{equation}
for any finite $\L$ and any cylinder function $f$ living in $\L$. Noticing that
$\varphi_{\L,f}$ is a bounded test function,
$\varphi_{\L,f}\in BC((W^{1,p}_0(\O)+u)\times BC(S)^\ast_+\times \O)$
we just have to verify that 
\begin{multline}
\nu_{\widetilde\Lambda,\e}(\varphi_{\L,f})=\int\int_{\O} \varphi_{\L,f}(\mu_{\widetilde\Lambda}(\cdot|\tau_{\lfloor{x}/{\e}\rfloor}(X))d\l(x)d\mu_{\e,u}(X)=\\=
\int\int_{\O}\bigl(\mu_{\widetilde\Lambda}(\mu_{\Lambda}(f | \cdot)|\tau_{\lfloor{x}/{\e}\rfloor}(X))-
\mu_{\widetilde\Lambda}(f|\tau_{\lfloor{x}/{\e}\rfloor}(X))\bigr)d\l(x)d\mu_{\e,u}(X)=0
\end{multline}
since $\mu_{\widetilde\Lambda}(\mu_{\Lambda}(f | \cdot)|\tau_{\lfloor{x}/{\e}\rfloor}(X))=
\mu_{\widetilde\Lambda}(f|\tau_{\lfloor{x}/{\e}\rfloor}(X))$ once $\L\subset\widetilde\Lambda$.

 To show that $\int \E_{\mu}(\nabla X(0))\,d \nu_{x,v}(\mu)=\nabla v(x)$, we use the test function
\begin{equation}
\varphi(v,\mu,x)= \eta(x) \E_{\mu}(\nabla X(0)-\nabla v(x))
\end{equation}
with $\eta\in C^1(\overline \O)^{m\times d}$ and observe that
\begin{multline}
\Bigl| \int \int \varphi(\Pie(X),\mu_{\widetilde\Lambda}(\cdot|\tau_{\lfloor{x}/{\e}\rfloor}(X)),x)d\l(x)d\mu_{\e,u}(X)\Bigr|\le\\
\le \Bigl| \int \eta(x) \int \bigl[ \nabla X(\lfloor{x}/{\e}\rfloor)-\nabla(\Pie(X))(x) \bigr] d\mu_{\e,u}(X) d\l(x)\Bigr|
\le\\
\le \text{const} \, \e \,\norm{\nabla\eta}_{\infty}\mu_{\e,u}(\norm{\nabla \Pie(X)}_1) .
\end{multline}
The last estimate is valid for the linear interpolation; for a more general case, $\norm{\nabla\eta}_{\infty}$
will be replaced by $\norm{\nabla^2\eta}_{\infty}$.

Finally, to show that $\nu_{v,x}(\abs{\nabla X(0)}^p) <\infty$, we can use the function
$\varphi_k(v,\mu,x)= \mu(\min(k, \abs{\nabla X(0)}^p))$ as a test function
yielding the bound uniform in $k$.
\qed

\appendix
\section{Technical Lemmas}

We begin with a technical Lemma that will be useful on several occasions.

\begin{lemma}\hfill
\label{L:tech}

Let $a>0$ and $\Lambda\subset \Oe$ be connected (when viewed as a subgraph of $\Z^d$ with the set of edges consisting of all pairs of nearest neighbours $(i,j), \abs{i-j}=1$). Then:

 a) We have
\begin{equation}
\label{E:bint}
\int \1_{\{j\},y}(X)\exp\bigl(-a\sum_{i\in \L}
\abs{\nabla X(i)}^p\bigr)\prod_{i\in\L} dX(i)\le  \o(m) \bigl(a^{-m/p}c(p,m)\bigr)^{\abs{\L}-1},
\end{equation}
where $j\in\L$ and  $\1_{\{j\},y}$ is the indicator of the set $\{X\in(\R^m)^{\L}\mid \abs{X(j)-y}<1\}$
and $\o(m)$ is  the volume of the unit ball in $\R^m$.

b) 
For any  $v\in L^r(\O,\R^m)$ and  $\e$ sufficiently small,
\begin{equation}
\int_{ {\mathcal N}_{\L,r}(v,\kappa)} \!\!\!\!\exp\bigl(-a\sum_{i\in \L} \abs{\nabla X(i)}^p \bigr)\prod_{i\in\L} dX(i)
\le  \vartheta \abs{\L}^{1+\frac{m}{d}} \bigl(a^{-m/p}c(p,m)\bigr)^{\abs{\L}-1},
\end{equation}
where $\vartheta=
\o(m )\kappa^m $ and 
$c(p,m)=\int_{\R^m} \exp\bigl(-\abs{\xi}^p\bigr)d\xi $.
\end{lemma}

\begin{proof}
a) Consider a tree $t$ rooted at  the site $j$. Then
\begin{equation}
\label{E:t}
-\sum_{i\in \L} \abs{\nabla  X(i)}^p\le - \sum_{\{i,j\}\in t}  |X(i)-X(j)|^p.
\end{equation}
and thus
\begin{equation}
\int \1_{\{j\},y}(X) \exp\bigl(-a\sum_{i\in \L} \abs{\nabla X(i)}^p\bigr)\prod_{i\in \L}dX(i)\le \o(m).
\Bigl(a^{-m/p}c(p,m)\Bigr)^{\abs{\L}-1}
\end{equation}

b) The set $\L$ is connected and can be covered by a spanning tree $t$ implying \eqref{E:t}. 
Further, we clearly have  ${\mathcal N}_{\L,r}(v,\kappa)\subset  \cup_{j\in \Oe} {\mathcal N}^{(j)}_{r}(v,\kappa)$ with
\begin{equation}
{\mathcal N}^{(j)}_{r}(v,\kappa)= \bigl\{X: \L \to \R^m\bigl|\abs{X(j)-X_{v,\e}(j)}< \kappa {\abs{\L}}^{\frac1d}\bigr\}.
\end{equation}
Considering the tree $t$ as rooted at $j$, we get
\begin{equation}
\int_{ {\mathcal N}^{(j)}_{r}(u,\kappa)} \exp\bigl(-a\sum_{i\in \L} \abs{\nabla X(i)}^p\bigr)\prod_{i\in \L}dX(i)\le 
\o(m )\kappa^m{\abs{\L}^{\frac{m}d}} \Bigl(a^{-m/p}c(p,m)\Bigr)^{\abs{\L}-1}
\end{equation}
implying the claim with the help of a).
\end{proof}

\begin{remark}
An immediate consequence of Lemma~\ref{L:tech} a), under the assumption (A1),  is the bound
\begin{equation}
\label{E:ZoeNr}
Z_{\Oe}(\mathcal N_{\Oe,r}(u,\kappa))\le   \vartheta \abs{\Oe}^{1+\frac{m}{d}}\bigl(c^{-m/p}c(p,m)\bigr)^{\abs{\Oe}-1}.
\end{equation}
\end{remark}

\begin{proofsect}{Proof of Exponential Tightness Lemma}

\noindent
The bound $H_{\Oe}(X)> K \abs{\Omega_{\e}}$ and (A1) implies that
\begin{equation}
- H_{\Oe}(X )< -\tfrac{1}2K  \abs{\Oe}-\tfrac12  c\!\!\!\sum_{\substack {j\in \Oe\\ \tau_j(A)\subset\Oe }} |\nabla X(j)|^p
\end{equation}
for all $X\in {\mathcal M}_K$. 
Hence,
\begin{multline}
\label{E:Nrkup}
\!\!\!\!\!\!Z_{\Oe}({\mathcal M}_K\cap {\mathcal N}_{\Oe,R_0,\infty}(X_{u,\e})) \le \\   \le
\exp\bigl(-\tfrac{1}2K \abs{\Oe}\bigr)\int_{ {\mathcal N}_{\Oe,R_0,\infty}(X_{u,\e})} \!\!\!\!\!\!\!\!\exp\bigl(-\tfrac12  c\!\!\!\sum_{\substack {j\in \Oe\\ \tau_j(A)\subset\Oe }} |\nabla X(j)|^p\bigr)
\prod_{i\in\Oe}dX(i) 
\end{multline}
Consider the set ${\O}^{0}_\e=\{j\in\Z^d|\tau_j(A)\subset \Oe\}$.
For sufficiently small $\e$, the set ${\O}^{0}_\e$ is connected, $\abs{\Oe\setminus   {\O}^{0}_\e}\le   C_{\partial} R_0 \e^{-d+1}\abs{\partial\O}$ and ${\O}^{0}_\e \cap S_{R_0}(\Oe)\neq\emptyset$. 
Thus, $X\in  {\mathcal N}_{\Oe,R_0,\infty}(X_{u,\e})$ implies  that $\abs{X(j)-X_{u,\e}(j)}\le 1$ for every $j\in{\O}^{0}_\e \cap S_{R_0}(\Oe)$.
Hence, using Lemma~\ref{L:tech} a) to bound the integral on the right hand side, 
we get
\begin{multline}
\label{E:Nrkuup}
\!\!\!\!\int_{ {\mathcal N}_{\Oe,R_0,\infty}(X_{u,\e})} \!\!\!\!\!\!\!\!\!\!\!\exp\bigl(-\tfrac12  c\!\!\!\sum_{\substack {j\in \Oe\\ \tau_j(A)\subset\Oe }} |\nabla X(j)|^p\bigr)
\prod_{i\in\Oe}dX(i)\le\\
 \le  \o(m) \Bigl[\bigl(\tfrac{2}{ c}\bigr)^{m/p}c(p,m)\Bigr]^{\abs{\Oe}}\bigl(\tfrac{\e_0}{\e}\bigr)^{m(1+\frac{d}r)C_{\partial} R_0 \e^{-d+1}\abs{\partial\O}}
\end{multline}
once $\e\le \e_0$ with 
$\e_0=\e_0(\kappa,r)=\bigl(\kappa\bigl(\frac{2}{c}\bigr)^{-1/p}c(p,m)^{-1/m}\bigr)^{\frac{r}{d+r}}\abs{\O}^{\frac1d}$. Thus
\begin{equation}
\label{E:ZMN<}
Z_{\Oe}({\mathcal M}_K\cap {\mathcal N}_{\Oe,r}(u,\kappa)) \le 
 \exp\bigl(-\tfrac{1}2K  \abs{\Oe}\bigr)   \Bigl[2\bigl(\tfrac{4}{c}\bigr)^{m/p}c(p,m)\Bigr]^{\abs{\Oe}}.
\end{equation}
Here,  we also used  the bound $\exp\bigl(-\tfrac{1}2K \abs{\Oe}\bigr) \bigl(\tfrac{\e_0}{\e}\bigr)^{m(1+\frac{d}r)C_{\partial} R_0 \e^{-d+1}\abs{\partial\O}}\le 1$
valid whenever 
\begin{equation}
 K>\e \abs{\log(\tfrac{\e}{\e_0})}  \,m(1+\tfrac{d}r)4C_{\partial}R_0 \frac{\abs{\partial\O}}{\abs{\O}}.
\end{equation}

For the second bound we derive 
\begin{multline}
\label{E:Nrkup}
 \!\!\!\!\!\!\!\!Z_{\Oe}({\mathcal M}_K\cap {\mathcal N}_{\Oe,r}(u,\kappa)) \le 
\exp\bigl(-\tfrac{1}2K \abs{\Oe}\bigr)\int_{ {\mathcal N}_{\Oe,r}(u,\kappa)} \!\!\!\!\!\!\!\!\exp\bigl(-\tfrac12  c\!\!\!\sum_{\substack {j\in \Oe\\ \tau_j(A)\subset\Oe }} |\nabla X(j)|^p\bigr)
\prod_{i\in\Oe}dX(i)
\end{multline}
in a similar way, using the fact that  $X\in  {\mathcal N}_{\Oe,r}(u,\kappa)$ implies  that $\abs{X(i)-\frac1\e u(\e i)}\le \kappa\abs{\Oe}^{\frac1r+\frac1d}$ for every $i\in \Oe$ 
and applying  Lemma~\ref{L:tech} b) to bound the integral on the right hand side.
We also assumed that $\e$ is sufficiently small so that $\vartheta \abs{\Oe}^{1+\frac{m}d} \le 2^{\abs{\Oe}}$.
\qed
\end{proofsect}

\begin{proofsect}{Proof of Interpolation Lemma}

Fixing parameters $\eta>0$ and  $N\in\N$, we   slice the strip $(\partial_\eta \O)_\e$
into strips of width $\frac{\eta}{\e N}$.
In particular, we pick up $R=R(\e)$ so that   $R>2R_0$ and $\eta=  N  \e R$
and partition the set   ${\mathcal N}_{\Oe,r}(u,\kappa)=\cup_{\ell=1}^{N-1}{\widehat{\mathcal N}}^{(\ell)}_{r}(u,\kappa)$
with 
\begin{equation}
{\widehat{\mathcal N}}^{(k)}_{r}(u,\kappa)=\{X\in {\mathcal N}_{\Oe,r}(u,\kappa)| \sum_{j\in S_{\!R_0 + k R}\setminus S_{\!R_0 +(k-1) R}} U(X_{\tau_j(A)})\le \tfrac1{N-1} H_{\Oe}(X)\},
\end{equation}
where $S_{\!R_0 +k R}=S_{\!R_0 +k R}(\Oe)$ is the strip  $S_{\!R_0 +k R}=\{i\in\Oe\colon \e^{-1}\dist(\e i,\O^{\co})\le R_0 +k R\}$.
To see that, indeed, ${\mathcal N}_{\Oe,r}(u,\kappa)\subset \cup_{k=1}^{N-1}{\widehat{\mathcal N}}^{(k)}_{r}(u,\kappa)$,
it suffices to show that any $X$ from  the set ${\mathcal N}_{\Oe,r}(u,\kappa)\setminus \cup_{k=1}^{N-1}{\widehat{\mathcal N}}^{(k)}_{r}(u,\kappa)$
would necessarily satisfy $H_{S_{NR}}(X)>H_{\Oe}(X)$ which is contradiction due to nonnegativity of $U(X_A)$. 
Further, introducing the function 
\begin{equation}
\varTheta_{k}(i)=\min(1, R^{-1} \bigl(\e^{-1}\dist(\e i,\O^{\co})- R_0-(k-1) R)_+)
\end{equation}
 on $\Oe$ interpolating between $1$ on $\Oe\setminus S_{\!R_0+k R}$  and $0$ on
$S_{\!R_0+(k-1) R}$,
we define, for any $X\in ({\R}^m)^{\Oe}$ and $Y\in ({\R}^m)^{S_{NR}}$ the function $T_{k}(X,Y)\in ({\R}^m)^{\Oe}$ by
\begin{equation}
\label{E:T}
T_{k}(X,Y)(i)=\varTheta_{k}( i) X(i)+(1-\varTheta_{k}( i))Y(i).
\end{equation}
It is interpolating between $T_{k}(X,Y)(i)=X(i)$ on $\Oe \setminus S_{\!R_0+k R}$ and $T_{k}(X,Y)(i)=Y(i)$ on $S_{\!R_0+(k-1) R}$.

Let $Z\in  {\mathcal N}_{\Oe,r}(u,\kappa)$ and consider 
 $X\in {\widehat{\mathcal N}}^{(k)}_{r}(u,\kappa)$ and $Y\in {\mathcal N}_{S_{ N R},\infty}(Z)$. For the completeness of the argument, let us first show that $T_{k}(X,Y)\in {\mathcal N}_{\Oe,r}(u,3\kappa)\cap{\mathcal N}_{\Oe,R_0,\infty}(Z)$ for each $1\le k\le N-1$. 
Indeed, extending $Y$ to $\Oe$ by taking  $Y(i)=Z(i) $  on $\Oe\setminus S_{ N R}$ and using that   $X,Z\in{\mathcal N}_{\Oe,r}(u,\kappa)$, we get
\begin{multline}
\label{E:TZ}
\norm{T_k(X,Y))-X_{u,\e}}_{\ell^r(\Oe)}\le \\ \le
\norm{\varTheta_k(X-X_{u,\e})}_{\ell^r(\Oe)}+\norm{(1-\varTheta_k)(Y-Z)}_{\ell^r(\Oe)}+\norm{(1-\varTheta_k)(Z-X_{u,\e})}_{\ell^r(\Oe)} \le\\
\le
\Bigl(2\kappa +\abs{S_{ N  R}}^{\frac1r}\abs{\Oe}^{-\frac1r-\frac1d}\Bigr) \abs{\Oe}^{\frac1r+\frac1d}
\le   3 \kappa \abs{\Oe}^{\frac1r+\frac1d}.
\end{multline}
Here, we first bounded
\begin{equation}
 \abs{S_{ N  R}}^{\frac1r}\abs{\Oe}^{-\frac1r-\frac1d}\le
  \bigl(C_{\p} N R\abs{\p\Omega}\e^{-(d-1)}\bigr)^{\frac1r}\abs{\O}^{-\frac1r-\frac1d}\e^{1+\frac{d}r}=
\tfrac{\e}{\abs{\O}^{1/d}}  \eta^{\frac1r} \bigl(\tfrac{C_{\p}\abs{\p\Omega}}{\abs{\O}}\bigr)^{\frac1r}
\end{equation}
with  $\eta=  N  R \e$  and  assumed that $\e$ is sufficiently small to assure that, with fixed $\eta$, the right hand side above does not exceed $\kappa$.  

The main idea of the proof is to introduce a new integral quantity that serves as an upper bound to the left hand side of 
\eqref{E:Zfr<Zb}  and, in the same time, as a lower bound of its right hand side.
To be more precise, for verification of an inequality of the form \eqref{E:Zfr<Zb} with the integral on the left hand side restricted to  ${\widehat{\mathcal N}}^{(\ell)}_{r}(u,\kappa)$,
we ``double the variables'' and introduce the following integral over $(\R^m)^{\Oe}\times (\R^m)^{S_{\!R_0+k R}}$,
\begin{equation}
\label{E:I}
 \!\!\!\!\!\!\!\!I_{k}=\int_{{\widehat{\mathcal N}}^{(\ell)}_{r}(u,\kappa)\times  {\mathcal N}_{S_{R_0+k R},\infty}({Z} )}  \mspace{-30mu}\exp\bigl(-H_{\Oe}(T_{k}(X,Y) )-a \!\!\!\sum_{j\in S_{\!R_0+k R}}  \abs{\nabla X(j)}^p\bigr) 
\prod_{i\in\Oe}dX(i)\prod_{j\in S_{\!R_0+k R}}dY(j).
\end{equation}

First, let us attend to the \emph {lower bound} on $I_{k}$. For the terms $U(T_{k}(X,Y)_{\tau_j(A)})$ contributing to $H_{\Oe}(T_{k}(X,Y))$ we consider 3 cases:

(i) If $\tau_j(A)\cap S_{\!R_0+k R}=\emptyset$, then $U(T_{k}(X,Y)_{\tau_j(A)})=U(X_{\tau_j(A)})$.

(ii) If $\tau_j(A)\cap (S_{\!R_0+k R}\setminus S_{\!R_0+(k-1) R}) \neq\emptyset$, then, by assumption (A2),
\begin{equation}
U(T_{k}(X,Y)_{\tau_j(A)})\le C\bigl(1+U(X_{\tau_j(A)}) +U(Y_{\tau_j(A)})+ \sum_{ i\in\tau_j(A)} \abs{\varTheta_{k}(i)-\varTheta_{k}(j)}^r \abs{X(i)-Y(i)}^r\bigr).
\end{equation}

In this inequality we used the fact that 
\begin{equation}
T_{k}(X,Y)(i)=\varTheta_{k}(j)X(i)+(1-\varTheta_{k}(j))Y(i) +(\varTheta_{k}(i)-\varTheta_{k}(j)) (X(i)-Y(i)).
\end{equation}

Again, for $i\notin S_{ N R}$, we have $Y(i)=Z(i)$.

(iii) If $\tau_j(A)\subset S_{\!R_0+(k-1) R}$, then $U(T_{k}(X,Y)_{\tau_j(A)})=U(Y_{\tau_j(A)})$.

The terms $ a  \abs{\nabla X(j)}^p$ in the integrand of $I_{k}$ are, according to (A1),   bounded as
\begin{equation}
\label{E:||U}
 a  \abs{\nabla X(j)}^p\le  U(X_{\tau_j(A)}) 
\end{equation}
for any $j\in S_{\!R_0+k R}$ and $a \le  c$.
As a result, using also the assumption (A2) in the form \eqref{E:A2short}, 
\begin{equation}
U(Y_{\tau_j(A)})\le  C\bigl(1+ U({Z}_{\tau_j(A)})  +\sum_{i\in  \tau_j(A)} \abs{Y(i) - {Z}(i)}^r\bigr)\le  C  (1+ U({Z}_{\tau_j(A)})+R_0^d ),
\end{equation}
we are getting the following bounds for the terms
\begin{equation}
L_j= U(T_{k}(X,Y)_{\tau_j(A)})+a\1_{j\in S_{\!R_0+k R}} \abs{\nabla X(j)}^p
\end{equation}
in the 3 cases as above:  

(i) $L_j\le U(X_{\tau_j(A)})$,

(ii) $L_j\le    \bigl( 1+C\bigr)U(X_{\tau_j(A)}) +C^2 (1+ U({Z}_{\tau_j(A)})+R_0^d )+  C+$ 
\smallskip

\rightline{$  +C\sum_{ i\in\tau_j(A)} \abs{\varTheta_{k}(i)-\varTheta_{k}(j)}^r \abs{X(i)-Y(i)}^r$,}

and

(iii)  $L_j\le  U(X_{\tau_j(A)})+C  (1+ U({Z}_{\tau_j(A)})+R_0^d ) $.

Then, for any $X\in  {\widehat{\mathcal N}}^{(\ell)}_{r}(u,\kappa)\setminus {\mathcal M}_K$ and with $a= c $ (and assuming that $C\ge 1$), we have
\begin{multline}
H_{\Oe}(T_{k}(X,Y) )+ c\sum_{i\in S_{\!R_0+k R}}   \abs{\nabla X(i)}^p\le  H_{\Oe}(X) + C \tfrac{K}N \abs{\O}\e^{-d} +\\
+ C^2\sum_{\substack{j\in S_{ N R} \\ \tau_j(A)\subset\Oe}}   U({Z}_{\tau_j(A)}) +
 \tilde c  \abs{S_{ N  R}} +   C R_0^{d+r} \bigl(\tfrac{ N}{\eta}\bigr)^r (2\kappa)^r \abs{\O}^{1+\frac{r}d}\e^{-d}
\end{multline}
with $\tilde c=C^2(1+C+R_0^d)+C$.
Here, for the last term, we used  the following estimate  with the  $\sum_j$  taken over all $j: \tau_j(A)\cap (S_{\!R_0+k R}\setminus S_{\!R_0+(k-1) R}) \neq\emptyset $, 
 \begin{multline}
 \label{E:(T-T)(X-Y)}
\sum_{j}\sum_{i\in\tau_j(A)} \abs{\varTheta_{k}(i)-\varTheta_{k}(j)}^r \abs{X(i)-Y(i)}^r\le\\  \!\!\!\!\!\!\!\!
\le  \bigl(\tfrac{R_0}{R}\bigr)^r R_0^d  \sum_{i\in S_{2R_0+k R}\setminus S_{(k-1) R} }\abs{X(i)-Y(i)}^r
\le   R_0^{d+r}R^{-r}  \bigl( \norm{X-Z}_{\ell^r(\Oe)} +\norm{Y-Z}_{\ell^r(S_{ N  R})}   \bigr)^r\le\\  \le
R_0^{d+r}R^{-r} \bigl(2\kappa \abs{\Oe}^{\frac1r+\frac1d}\bigr)^r=R_0^{d+r} \bigl(\tfrac{ N}{\eta}\bigr)^r \e^{-d}(2\kappa)^r \abs{\O}^{1+\frac{r}d}.
\end{multline}
To get this, we first used that $\abs{\varTheta_{k}(i)-\varTheta_{k}(j)}\le \tfrac{R_0}R$ for any $i\in\tau_j(A)$ since $\diam A<R_0$
and then  applied the bound from \eqref{E:TZ} assuming that     $\e$ is sufficiently small (in dependence on $\kappa$ and $\eta$).
As a result,
we get 
\begin{equation}
 \!\!\!\!\!\!\!\! H_{\Oe}(T_{k}(X,Y) )+ c\sum_{i\in S_{\!R_0+k R}}   \abs{\nabla X(i)}^p\le  H_{\Oe}(X) +\tilde C\bigl( (\tfrac{K}N +
\eta + (\tfrac{N \kappa}\eta )^r) \e^{-d} +\sum_{\substack{j\in S_{ N R} \\ \tau_j(A)\subset\Oe}}   U({Z}_{\tau_j(A)})\bigr)
\end{equation}
with 
\begin{equation}
\tilde C = \max\bigl(C \abs{\O},\tilde c  C_{\p} \abs{\p \O}, R_0^{d+r}  2^r \abs{\O}^{1+\frac{r}d},C^2\bigr).
\end{equation} 
Thus, finally,
\begin{equation}
\label{e:ubIj}
 \!\!\!\!\!\!\!\!
Z_{\Oe}( {\widehat{\mathcal N}}^{(\ell)}_{r}(u,\kappa)\setminus {\mathcal M}_K)
\exp\bigl\{- \tilde C\bigl((\tfrac{K}N +\eta + (\tfrac{N\kappa}\eta)^r)) \e^{-d}+ 
\sum_{\substack{j\in S_{ N R} \\ \tau_j(A)\subset\Oe}}   U({Z}_{\tau_j(A)})\bigr)\bigr\}\o(m)^{\abs{S_{R_0+k R}}} \le I_{k}.
\end{equation}

For the \emph{upper bound} of the integral $I_{k}$, we use the substitution defined as identity  on $(\R^m)^{\Oe\setminus S_{\!R_0+k R}}$, and, on the remaining $(\R^m)^{S_{\!R_0+k R}}\times  (\R^m)^{S_{\!R_0+k R}}$, pointwise by the mapping $\Phi_i:\R^m\times \R^m\to \R^m\times \R^m$ introduced by
\begin{equation}
\Phi_i(\xi,\zeta)=(\varTheta_{k}(i)\xi+(1-\varTheta_{k}(i))\zeta,\zeta)  \text{ for }  i\in S_{\!R_0+k R}\setminus S_{\!R_0+(k-1/2) R}
\end{equation}
and by
\begin{equation}
\Phi_i(\xi,\zeta)=(\varTheta_{k}(i)\xi+(1-\varTheta_{k}(i))\zeta,\xi)  \text{ for }  i\in S_{\!R_0+(k-1/2) R}.
\end{equation}
Notice that    
\begin{equation}
\abs{\det D\Phi_i}^{-1} \le 2^m \1_{i\in S_{\!R_0+k R}\setminus S_{\!R_0+(k-1) R}}+ \1_{i\in  S_{\!R_0+(k-1) R}}.
\end{equation}

Since $T_{k}(X,Y)\in {\mathcal N}_{\Oe,r}(u,2\kappa)\cap{\mathcal N}_{\Oe,R_0,\infty}(Z)$, we have
\begin{multline}
\label{E:lbIj}
\!\!\!\!I_{k}\le Z_{\Oe} ({\mathcal N}_{\Oe,r}(u,2\kappa)\cap{\mathcal N}_{\Oe,R_0,\infty}(Z))  \omega(m)^{\abs{S_{\!R_0+k R}\setminus S_{\!R_0+(k-1/2)R}}}
2^{m\abs{S_{\!R_0+k R}\setminus S_{\!R_0+(k-1) R}}}\times \\ \times
\bigl(2c^{-\frac{m}p} c(p,m)\bigr)^{\abs{S_{\!R_0+(k-1/2)R}}}.
\end{multline} 
Here, the last factor  arises as  the bound on the integral
\begin{equation}
\int_{ {\mathcal N}_{S_{\!R_0+(k-1/2)R},r}(u,\kappa)}   \exp\bigl\{-c \sum_{i\in S_{\!R_0+(k-1/2)R}} \abs{\nabla  X(i)}^p\bigr\}
\prod_{i\in S_{\!R_0+(k- 1/2)R}}dX(i)
\end{equation}
according to Lemma~\ref{L:tech} a)
with
\begin{equation}
\vartheta\abs{S_{\!R_0+(k-\frac12)R}}^{1+m/d}\bigl(c^{-\frac{m}p} c(p,m)\bigr)^{\abs{S_{\!R_0+(k-1/2)R}}-1}\le 
\bigl(2c^{-\frac{m}p} c(p,m)\bigr)^{\abs{S_{\!R_0+(k-1/2)R}}}
\end{equation}
valid for $\e$ sufficiently small (estimating $\abs{S_{\!R_0+(k-1/2)R}}\ge \abs{S_{\!R_0+\frac12 R}}\sim \e^{-d+1} R=\frac{\eta}N\e^{-d}$.

Combining \eqref{e:ubIj} with \eqref{E:lbIj} for each of $N-1$ integrals over ${\widehat{\mathcal N}}^{(\ell)}_{r}(u,\kappa)\setminus {\mathcal M}_K$,
we get 
\begin{multline}
\label{E:equivfin}
Z_{\Oe} ({\mathcal N}_{\Oe,r}(u,\kappa)\setminus {\mathcal M}_K)\le\\ \le  \exp\bigl\{\mathcal C \bigl((\tfrac{K}N +\eta + (\tfrac{N\kappa}\eta)^r)) \e^{-d}+ \sum_{\substack{j\in S_{ N R} \\ \tau_j(A)\subset\Oe}}  U({Z}_{\tau_j(A)})\bigr) \bigr\} \
Z_{\Oe} ({\mathcal N}_{\Oe,r}(u,2\kappa)\cap{\mathcal N}_{\Oe,R_0,\infty}(Z)) .
\end{multline}
Here we bounded the prefactor $N-1$ (the number of terms with $k=1,\dots,N-1$) combined with the factors in  \eqref{E:lbIj} by
\begin{equation}
 N\bigl(\omega(m) 2^{m+1}c^{-\frac{m}p} c(p,m)\bigr)^{\abs{S_{ N R}}}\le e^{\frac23\mathcal C\eta \e^{-d}}
\end{equation}
with a constant $\mathcal C=3\max\bigl(\tilde C  ,  C_{\p} \abs{\p\O} \log(\omega(m) 2^{m+1}c^{-\frac{m}p} c(p,m)) \bigr)$.
We used  the bound
$\abs{S_{ N R}}\le C_{\p} \abs{\p\O}\e^{-d+1}  N R= C_{\p} \abs{\p\O}\e^{-d} \eta$ and bounded,  for  $\e$ sufficiently small, the term $N=e^{(\e^d\log N)\e^{-d}}$
by $\exp\{\frac13\mathcal C \eta \e^{-d}\}$.

According to  Lemma~\ref{L:tight}, we have
$Z_{\Oe}({\mathcal M}_K\cap {\mathcal N}_{\Oe,r}(u,\kappa))   \le
e^{-\frac12 K \abs{\Oe}} D^{\abs{\Oe}}  $.
Hence,
\begin{equation}
\label{E:K12}
Z_{\Oe} ({\mathcal N}_{\Oe,r}(u,\kappa)\cap {\mathcal M}_K)\le    \frac12 Z_{\Oe} ({\mathcal N}_{\Oe,r}(u,\kappa))
\end{equation}
once we choose $K=  \log D + \e^d \frac{\log 2}{\abs{\O}} + F_{\kappa,\e}(u)$.
Then, multiplying the right hand side in  \eqref{E:equivfin} by 2, we can replace $Z_{\Oe} ({\mathcal N}_{\Oe,r}(u,\kappa)\setminus {\mathcal M}_K)$  by $Z_{\Oe}({\mathcal N}_{\Oe,r}(u,\kappa))$, yielding the claim with a slight increase of $\mathcal C$ and with 
$b=1+\log D$ for sufficiently small $\e$.
 \qed
\end{proofsect} 

\noindent
\textbf{Acknowledgement}
The research of R.K. was supported  by the grants GA\v CR 201-09-1931, 201/12/2613, and MSM 0021620845.
Both authors were also supported by FG \emph{Analysis and Stochastics in Complex Physical Systems} and Hausdorff Research Institute for Mathematics.


\begin{thebibliography}{1}


\bibitem{DGI}
J.-D.~Deuschel, G.~Giacomin, and D.~Ioffe,
 \emph{Large deviations and concentration properties
for $\nabla\varphi$ interface models.}
Probability Theory and Related Fields {\bf 117}, 49--111 (2000).

\bibitem{S}
S.~Sheffield,
\emph{Random surfaces.}
Ast\'erisque {\bf 304}, 175 pages, (2005).



\bibitem{B}
A.~Braides,
\emph{Gamma-Convergence for Beginners.}
Oxford Lecture Series in Mathematics and Its Applications {\bf 22}
224 pages (2002).

\bibitem{BLM}
A.~Bourgeat, S.~Luckhaus, and A.~Mikeli\'c,
\emph{Convergence of the  homogenization process for a double-porosity model of immiscible two-phase flow.}
SIAM J. Math. Anal.	{\bf 27},1520--1543 (1996).



\bibitem{FS}
T.~Funaki and H.~Spohn,  
\emph{Motion by Mean Curvature from the Ginzburg-Landau $\nabla\varphi$ Interface Model.}
 Communications in Mathematical Physics {\bf 185 }, 1--36 (1997).

\bibitem{BK}
M.~Biskup and  R.~Koteck{\'y}.
\emph{Phase coexistence of gradient Gibbs states.}
Probability Theory and Related Fields
\textbf{139}, 1--39 (2007).

\bibitem{FT}   
G.~Friesecke and F.~Theil.
\emph{Validity and Failure of the Cauchy-Born Hypothesis in a Two-Dimensional Mass-Spring Lattice.}
Journal of Nonlinear Science {\bf 12}, 445--478 (2002).





\end{thebibliography}
\end{document}